\pdfoutput=1
\documentclass[format=sigconf,screen=true,10pt,nonacm=true]{acmart}

\usepackage[T1]{fontenc}
\usepackage{etoolbox,graphicx,setspace,listings,multicol,xspace,enumitem,booktabs,color}
\usepackage{subcaption}
\usepackage{algorithm,mathtools}
\usepackage[noend]{algpseudocode}
\usepackage[skip=10pt,labelfont=bf]{caption}

\usepackage{siunitx}
\newtoggle{arxiv}
\toggletrue{arxiv}
\iftoggle{arxiv}{
    \settopmatter{printacmref=false, printccs=false, printfolios=true}
    \renewcommand\footnotetextcopyrightpermission[1]{} % removes footnote with conference information in first column
    \pagestyle{plain} % removes running headers
    \setcopyright{none} 
    \usepackage{mathtools}
    \usepackage{amsmath,amsthm}
}{
	\vldbTitle{Storyboard: Optimizing Precomputed Summaries for Aggregation}
	\vldbAuthors{Edward Gan, Peter Bailis, and Moses Charikar}
	\vldbDOI{https://doi.org/10.14778/xxxxxxx.xxxxxxx}
	\vldbVolume{12}
	\vldbNumber{xxx}
	\vldbYear{2019}
	\usepackage[hyphens]{url}
}

\newtheorem{theorem}{Theorem}
\newtheorem{lemma}{Lemma}

\newtheorem{cor}{Corollary}

\iftoggle{arxiv}{
	\newtheorem*{defn*}{Definition}
	\newtheorem*{lemma*}{Lemma}
	\newtheorem{defn}{Definition}
}{
	\newdef{defn}{Definition}
}

\newcommand{\minihead}[1]{{\vspace{.45em}\noindent\textbf{#1.} }}
\newcommand{\ssep}{:}
\newcommand{\DD}{\mathcal{D}}
\newcommand{\QQ}{\mathcal{Q}}
\newcommand{\pot}{\phi}
\newcommand{\ea}{\varepsilon}
\newcommand{\er}{{\epsilon'}}
\newcommand{\pre}{\text{Pre}}

\newcommand*{\preagg}{{AggPre}\xspace}
\newcommand*{\coopf}{\texttt{Coop\-Freq}\xspace}
\newcommand*{\coopq}{\texttt{Coop\-Quant}\xspace}

\newcommand*{\sampling}{\texttt{USample}\xspace}
\newcommand*{\dyadic}{\texttt{Hier\-archy}\xspace}
\newcommand*{\cms}{\texttt{CMS}\xspace}
\newcommand*{\truncation}{\texttt{Trun\-cation}\xspace}

\newcommand*{\kll}{\texttt{KLL}\xspace}
\newcommand*{\sbopt}{\texttt{SB}\xspace}
\newcommand*{\samplingp}{\texttt{USample:Prop}\xspace}
\newcommand*{\strat}{\texttt{STRAT}\xspace}

\newcommand*{\uniform}{\texttt{Uni\-form}\xspace}
\newcommand*{\zipf}{\texttt{Zipf}\xspace}

\newcommand*{\caida}{\texttt{CAIDA}\xspace}
\newcommand*{\insta}{\texttt{Insta\-cart}\xspace}
\newcommand*{\power}{\texttt{Power}\xspace}
\newcommand*{\servicet}{\texttt{Tra\-ffic}\xspace}
\newcommand*{\serviceo}{\texttt{OS\-Build}\xspace}
\newcommand*{\servicep}{\texttt{Pro\-vider}\xspace}

\newcommand*{\storyboard}{Story\-board\xspace}
\newcommand*{\msft}{Micro\-soft\xspace}
\newcommand*{\imply}{Imply\xspace}
\newcommand*{\pps}{PPS\xspace}

\AtBeginDocument{%
  \providecommand\BibTeX{{%
    \normalfont B\kern-0.5em{\scshape i\kern-0.25em b}\kern-0.8em\TeX}}}

\begin{document}

\iftoggle{arxiv}{
	\title{Storyboard: Optimizing Precomputed Summaries for Aggregation}
	\titlenote{Preprint. Under review.}
	\author{Edward Gan, Peter Bailis, Moses Charikar}
	\affiliation{
	    \institution{Stanford University}
	}
	\begin{abstract}
An emerging class of data systems partition their data and precompute approximate summaries (i.e., sketches and samples) for each segment to reduce query costs. They can then aggregate and combine the segment summaries to estimate results without scanning the raw data. However, given limited storage space each summary introduces approximation errors that affect query accuracy. For instance, systems that use existing mergeable summaries cannot reduce query error below the error of an individual precomputed summary. We introduce Storyboard, a query system that optimizes item frequency and quantile summaries for accuracy when aggregating over multiple segments. Compared to conventional mergeable summaries, Storyboard leverages additional memory available for summary construction and aggregation to derive a more precise combined result. This reduces error by up to 25$\times$ over interval aggregations and 4.4$\times$ over data cube aggregations on industrial datasets compared to standard summarization methods, with provable worst-case error guarantees.
\end{abstract}
	\maketitle
}{
	\title{Storyboard: Optimizing Precomputed Summaries for Aggregation}
	\numberofauthors{3}
	\author{
	  Edward Gan, Peter Bailis, Moses Charikar \\
	  \affaddr{Stanford University}\\
	  \affaddr{Stanford, CA, USA}
	}
	\maketitle
	\begin{abstract}
An emerging class of data systems partition their data and precompute approximate summaries (i.e., sketches and samples) for each segment to reduce query costs. They can then aggregate and combine the segment summaries to estimate results without scanning the raw data. However, given limited storage space each summary introduces approximation errors that affect query accuracy. For instance, systems that use existing mergeable summaries cannot reduce query error below the error of an individual precomputed summary. We introduce Storyboard, a query system that optimizes item frequency and quantile summaries for accuracy when aggregating over multiple segments. Compared to conventional mergeable summaries, Storyboard leverages additional memory available for summary construction and aggregation to derive a more precise combined result. This reduces error by up to 25$\times$ over interval aggregations and 4.4$\times$ over data cube aggregations on industrial datasets compared to standard summarization methods, with provable worst-case error guarantees.
\end{abstract}
}
\section{Introduction}
\label{sec:intro}
% Context: introduce AggPre with summaries
An emerging class of data systems precompute aggregate summaries over a dataset to reduce query times.
These precomputation (\preagg~\cite{peng2018aqp}) systems trade off preprocessing
time at data ingest to avoid scanning the data at query time.
In particular, Druid and similar systems partition datasets into disjoint segments 
and precompute summaries for each segment \cite{yang2014druid,IM2018Pinot}.
They can then process queries by aggregating results from the segment summaries.
Unlike traditional data cube systems \cite{gray1997datacube}, the summaries 
go beyond scalar counts and sums and include data structures that can
approximate quantiles and frequent items \cite{cormode2012synopses}.
As an example, our collaborators at \msft often issue queries to estimate 99th percentile
request latencies over hours-long time windows.
Their Druid-like system precomputes quantile summaries \cite{dunning2019tdigest,gan2018moments} 
for 5 minute time segments and then combines summaries to estimate quantiles
over a longer window,
reducing data access and runtime at query time by orders of magnitude \cite{druidblog}.

% Need: queries are limited by precomputed summary accuracy due to system assumptions
Although querying summaries is more efficient than querying raw data,
precomputing summaries also limits query accuracy.
Given a total storage budget and many data segments, each
segment summary in an \preagg system has limited storage space --
often <$10$ kilobytes -- and thus limited accuracy \cite{druidquantilesketch}.
Prior work on \emph{mergeable summaries} introduces summaries that can be combined with
no loss in accuracy, and are commonly used in \preagg systems \cite{agarwal2012mergeable,gan2018moments,druidblog}.
However, even mergeable summaries have maximum accuracy capped by the 
accuracy of an individual summary.
\begin{figure}
\includegraphics[width=\columnwidth]{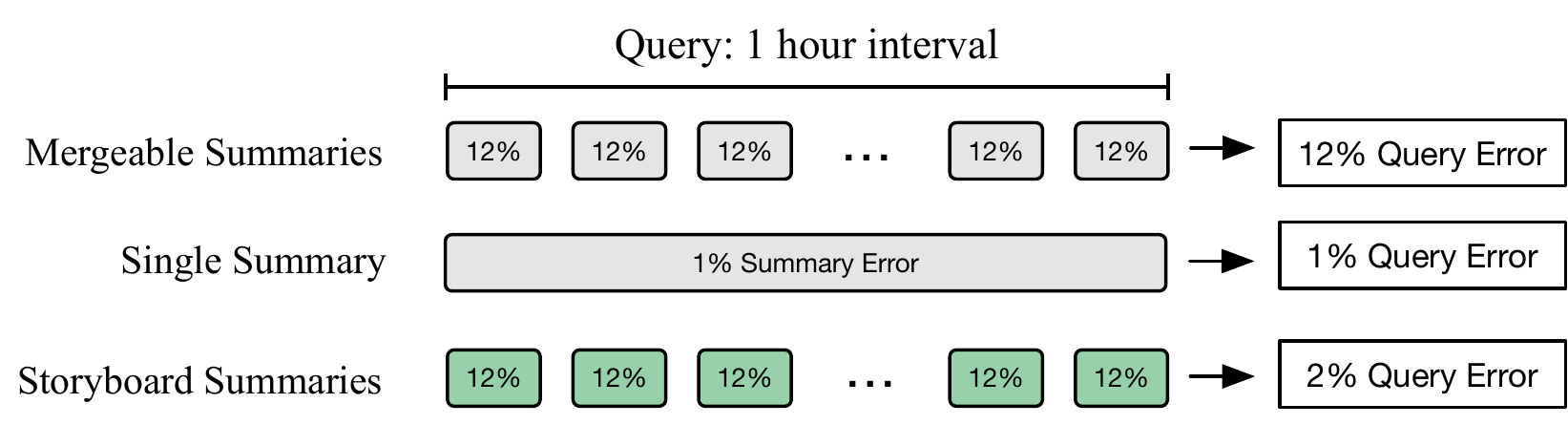}
\caption{Given a space budget, mergeable summaries preserve accuracy when combined but cannot match the accuracy of using a single larger summary. \storyboard closes the gap by
optimizing summaries for accurate aggregations.}
\label{fig:error_tradeoff}
\vspace{-1em}
\end{figure}
We illustrate this challenge in Figure~\ref{fig:error_tradeoff}.
Consider a query for the 99th percentile latency from \texttt{1:05pm} to \texttt{2:05pm}, 
and suppose we precompute mergeable quantile summaries for 5 minute time segments 
that individually have $12\%$ error.
Calculating quantiles over the full hour requires aggregating results from $12$ summaries,
and mergeable summaries would maintain $12\%$ error for the final result.
This is not ideal: if the same space were instead used to store a single large summary for the
entire interval, we would have $12\times$ less error with $\epsilon = 1\%$.
On the other hand, using a single large summary restricts the granularity of
possible queries.

% Object: What is Storyboard
In this paper we introduce \storyboard, an \preagg query system
that optimizes frequent items and quantile summaries for aggregation.
Unlike mergeable summaries, \storyboard queries that combine results from multiple summaries have
lower relative error than any summary individually.
To do so \storyboard uses a different resource model than most existing summaries were designed for.
While mergeable summaries assume the amount of memory available for
constructing and aggregating (combining) summaries is the same as that for storage, we
have seen in real-world deployments that \preagg systems have orders of magnitude more
memory for construction and aggregation.
\storyboard takes advantage of these additional resources to construct summaries that
compensate for the errors in other summaries they may be combined with,
and then aggregate results using a large and precise accumulator.

\storyboard focuses on supporting queries over intervals and data cubes roll-ups.
Interval queries aggregate over one-dimensional contiguous ranges, such as a time window from \texttt{1:00pm} to \texttt{9:00pm} \cite{basat2018interval}, while
data cube queries aggregate over data matching specific dimension values,
such as \texttt{loc=USA AND type=TCP} \cite{gray1997datacube}.
When aggregating $k$ summaries together, \storyboard can reduce relative error by
nearly a factor of $k$ for interval aggregations and a factor of $\sqrt{k}$ for
other aggregations, compared with no reduction in error for mergeable summaries.
These two query types cover a wide class of common queries and \storyboard can construct
summaries optimized for either of the two types.

% Time Interval Queries
\minihead{Interval Queries}
For interval queries, \storyboard uses novel summarization techniques which we
call \emph{cooperative} summaries.
Cooperative summaries account for the cumulative error over consecutive sequences of summaries,
and adjust the error in new summaries to compensate.
For instance, if five consecutive item frequency (heavy hitters) 
summaries have tended to underestimate the true frequency of
item $x$, cooperative summaries can bias the next summary to overestimate $x$.
This keeps the total error for queries spanning $k$ segments smaller than existing summarization techniques.
Hierarchical approximation techniques \cite{basat2018interval,qadarji2013hierarchical}
can also be used here but require additional space and provide worse accuracy
in practice.

We prove that our summaries have cumulative error no worse than state of the art randomized
summaries \cite{huang2011optimalsamp}, 
and for frequencies exceed the accuracy of state of the art 
hierarchical approaches \cite{basat2018interval}.
Empirically, cooperative summaries provide a 4-25$\times$ reduction in error on queries
aggregating multiple summaries compared with existing sketching and summarization techniques.

% Challenge 2: Multi-dimensional
\minihead{Multi-dimensional Cube Queries}
Data cube queries can aggregate the same summary along different dimensions,
so compensating for errors explicitly along a single dimension is insufficient.
Instead, for cube workloads \storyboard used randomized weighted samples 
and reduces error further by optimizing for an expected workload of queries.
\storyboard exploits the fact that data cubes often have dimensions with skewed value distributions: 
some values or combination of values that occur far more frequently than others. 
Then, \storyboard optimizes the allocation of storage space and introduces targeted biases 
where they will have the greatest impact to minimize average query error.
Empirically, these optimizations yield an up to $4.4\times$ reduction in average error 
compared with standard data cube summarization techniques.

% Conclusion
In summary, we make the following contributions:
\begin{enumerate}
	\setlength\itemsep{.5em}
	\setlength{\topsep}{0em} 
	\item We introduce Storyboard, an approximate \preagg system that provides improved
	query accuracy for aggregations by taking advantage of additional memory resources
	at data ingest and query time.
	\item We develop cooperative frequency and quantile summaries that minimize error when aggregating over intervals and establish worst-case bounds on their error.
	\item We develop techniques for allocating space and bias among randomized summaries to minimize average error under cube aggregations. 
\end{enumerate}

% Outline
The remainder of the paper proceeds as follows.
In Section~\ref{sec:sys_motivation} we provide motivating context.
In Section~\ref{sec:system} we present \storyboard and its query model. 
In Section~\ref{sec:timerange} we describe cooperative summaries optimized for intervals.
In Section~\ref{sec:optcube} we describe optimizations for data cubes. 
In Section~\ref{sec:eval} we evaluate \storyboard's accuracy. 
We describe related work in Section~\ref{sec:related} and conclude in Section~\ref{sec:conclusion}.
\section{Background}
\label{sec:sys_motivation}
\storyboard targets aggregation queries common in monitoring 
and data exploration applications, and improves
accuracy for these queries by taking advantage of memory resources that
real world summary precomputation systems already have available at query and ingest time.
This motivation comes out of our experience collaborating with engineers at \imply,
the developers of Druid, as well as a cloud services team at \msft.

% Real-world aggregation usage
\minihead{Aggregation Queries}
% TODO: fix graph legend label
Existing usage of \preagg systems feature queries that
span many segment summaries but follow structured patterns.
At a cloud services team at \msft, users interacted with a Druid-like system
primarily through a time-series monitoring dashboard
which they used to track trends and explore anomalies.
In addition to count and average queries,
Top-K (heavy hitters) item frequency and quantile queries were prevalent.
The most common forms of
aggregations were over time intervals and grouped cube roll-ups.
For instance, engineers often wanted to see the most common ip address frequencies
over specific release windows or server configurations.
\begin{figure}
\includegraphics[width=\columnwidth]{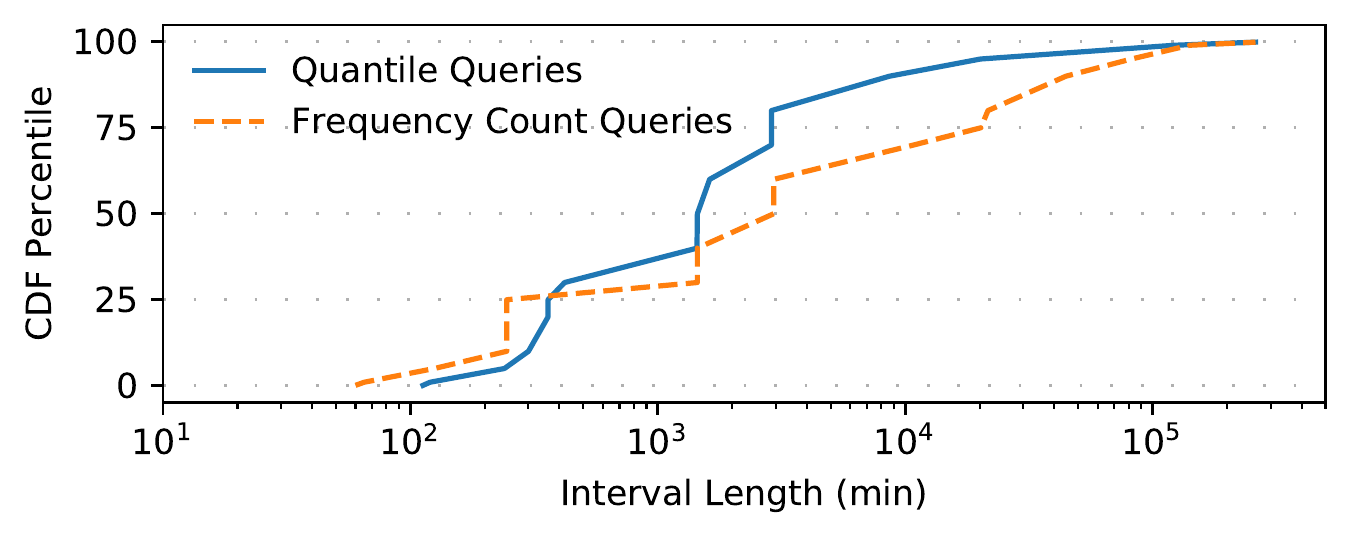}
\caption{Distribution of user-issued time interval queries to a Druid-like system at \msft. More than half of the queries span $>100$ five-minute segments.}
\label{fig:intervals}
\end{figure}

Notably, the queries over time intervals and data cubes involved 
combining results from a large number of summaries.
In Figure~\ref{fig:intervals} we describe a set of 
33K Top-K item frequency queries and 130K quantile queries issued 
to an \preagg system at \msft.
More than half of the queries span intervals longer than a day.
Since the system stored summaries at a 5 minute granularity,
this meant combining results from hundreds of summaries or
reverting to less accurate, coarser grained roll-ups.
Data cube roll-ups were also common. 
Users would commonly issue queries that filtered or grouped on a dimension,
and ``drill-down'' \cite{gray1997datacube} as needed, 
sometimes as part of anomaly explanation systems like MacroBase \cite{abuzaid2018diff,abuzaid2018macrobase}.
Over 50\% of all cubes had more than 10K dimension value segments and queries
that span hundreds of cube segments were common.

\storyboard thus optimizes for the accuracy of time interval and cube queries that span not just a single summary, but aggregate over many summaries.

% Real-world system constraints
\minihead{Memory Constraints}
At both \imply and \msft, the storage space available to each summary was limited.
Since summaries are maintained for each of potentially millions of segments,
and nodes have limited memory and cache, the memory must be divided amongst the segment summaries. 
To illustrate, by default each quantile summary in Druid is configured for $2\%$ error and roughly 10 kB of memory usage \cite{druidquantilesketch}.
Thus it is important to use summaries that provide high accuracy with minimal storage overhead.
However, in real-world deployments there is much more memory available during summary construction and aggregation than there is for storage.
For instance, in Druid the use of Hadoop map-reduce jobs for batch data ingestion
allows for effectively unconstrained memory limits during summary construction.
Then, at query time, engineers at \imply report that a standard
deployment uses query processing jobs with $0.5$ GB of memory each when aggregating
summaries together.

While standard mergeable summaries \cite{agarwal2012mergeable} are designed to maintain
the same memory footprint during construction and aggregation that they use in storage,
\storyboard takes advantage of additional memory at both construction and query time to achieve higher query accuracy.
\section{System Design}
\label{sec:system}
\label{sec:sys-arch}
In this section we describe \storyboard's system design.
We discuss the types of queries supported,
how summaries are constructed for different query types,
and how summaries are aggregated to provide accurate query results.
We outline the system components in Figure~\ref{fig:sysoverview}.
\begin{figure}
\includegraphics[width=\columnwidth]{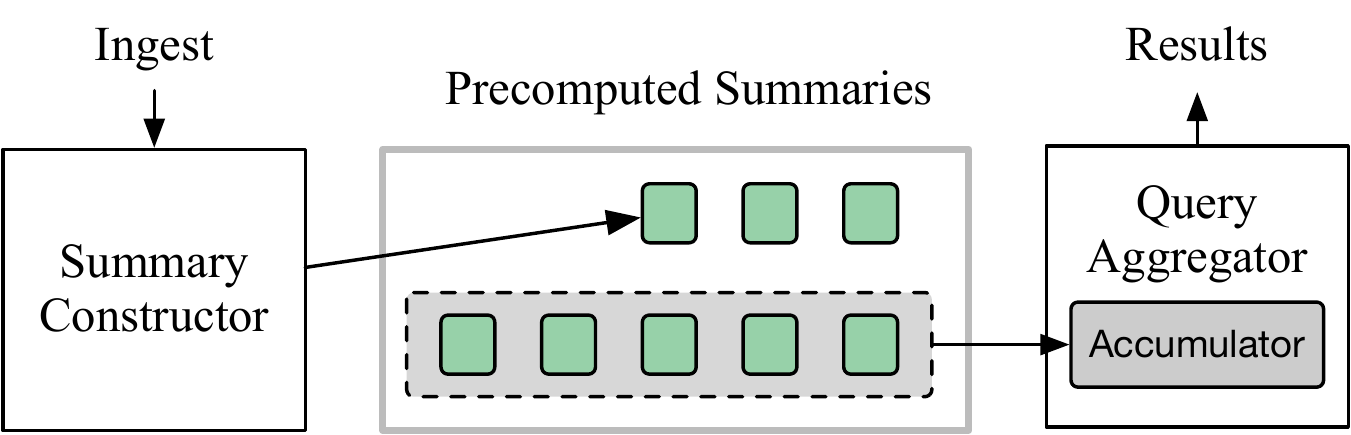}
\caption{\storyboard precomputes summaries at ingest optimized to minimize error under
aggregations. At query time, results from multiple summaries are combined using a precise
accumulator to provide accurate results.}
\label{fig:sysoverview}
\iftoggle{arxiv}{}{\vspace{-1em}}
\end{figure}

\subsection{Queries}
\label{sec:sys-queries}
% Data spec
Consider data records
$\rho = (x, t, d_{1}, \dots, d_{m_{d}})$
where $x$ is either a categorical or ordinal value of interest (i.e. ip address, latency), 
$t$ is an ordered dimension for interval queries (i.e. timestamp), 
and the $d_{j}$ are categorical dimensions (i.e. location).
A \storyboard query $g_{\QQ}(x)$ specifies an aggregation of records $\QQ$
and a function $g$ to estimate for the value $x$.
$\QQ$ defines an \emph{aggregation} with a selection condition:
either a one-dimensional interval or a
multi-dimensional cube query \cite{basat2018interval,gray1997datacube}.
\begin{defn}
An interval aggregation specifies
$
\QQ^{\text{(time)}} = \{\rho \ssep T_0 \leq t < T_1\}
$
for $T_0,T_1$ aligned at a time-resolution $T_G$ ($T_0,T_1 \bmod T_G = 0$) and maximum length $T_1 - T_0 \leq k_T\cdot T_G$.
\label{defn:interval}
\end{defn}
\begin{defn}
A data cube aggregation specifies
$
\QQ^{\text{(cube)}} = \{\rho \ssep d_{i_1} = v_{i_1} \wedge \dots\ \wedge d_{i_k} = v_{i_k} \}
$
for $d_{i_1},\dots,d_{i_k}$ a subset of the dimensions to condition on.
\label{defn:cube}
\end{defn}
The \emph{query function} $g$ is either an item frequency $f$ or rank $r$ \cite{cormode2010frequent,luo2016quantiles}.
An item frequency $f(x)$ is the total count of records with value $x$
while a rank $r(x)$ is the total count of records with values less than or equal to $x$.
We use $g$ generically denote either frequencies or ranks.
\begin{equation}
f_{\QQ}(x) = \sum_{\rho_i \in \QQ} 1_{x_i = x} \qquad\qquad
r_{\QQ}(x) = \sum_{\rho_i \in \QQ} 1_{x_i \leq x}.
% \begin{aligned}
% f_{\QQ}(x) &= \sum_{\DD_i \in \QQ} \DD_i(x) \\
% r_{\QQ}(x) &= \sum_{\DD_i \in \QQ} \sum_{y \in \DD_i} \DD_i(y)\cdot 1_{y \leq x}.
% \end{aligned}
\label{eqn:queryfunc_defs}
\end{equation}
Using these primitives, \storyboard can also return estimates for quantiles
and Top K / Heavy Hitters queries, which we will discuss in more detail in Section~\ref{sec:sys_query_proc}.
% Total counts, sums, and averages can be incorporated
% with minimal overhead using standard statistics \cite{harinarayan1996cubes,wasay2017canopy}
% so we do not discuss them further.
% Summaries for distinct items queries such as HyperLogLog~\cite{Flajolet07hyperloglog} can also
% be incorporated as in Druid but our advanced optimizations would not apply since 
% distinct counts are not additive.

\subsection{Data Ingest}
\label{sec:sys_construction}
Before \storyboard can ingest data, users specify whether they want the
dataset to support interval or data cube aggregations, and whether they want
the dataset to support frequency or rank query functions.
Users also specify total space constraints and workload parameters.
A dataset can be loaded multiple times to support different combinations of the above.

\storyboard then splits the data records into atomic segments $\DD$.
These segments form a disjoint partitioning of a dataset, and are chosen so that
any aggregation can be expressed as a union of segments.
For interval aggregations users specify a time resolution $T_G$ and a maximum length $k_T$,
defining segments $\DD_i = \{\rho : i \cdot T_G\leq t < (i+1)\cdot T_G\}$
For cube aggregations the partitions are defined by grouping by all 
$m_d$ of the dimensions 
$\DD_{\vec{v}} = \{\rho : d_1 = v_1 \wedge \dots \wedge d_{m_d} = v_{m_d}\}$.
Once the dataset is partitioned we can represent the records in each segment
as mappings from values to counts: 
$$
\DD = \{x_1 \mapsto \delta_1, \dots, x_r \mapsto \delta_r\}.
$$

For each segment $\DD$ \storyboard constructs a summary $S$ 
consisting of $s$ value, count mappings 
$$S \coloneqq \{x_1 \mapsto \gamma_1, \dots, x_s \mapsto \gamma_s\}.$$ 
This is similar to other counter based summaries 
\cite{metwally2005efficient,cormode2010frequent} 
and weighted sampling summaries \cite{huang2011optimalsamp}.
Unlike tabular sketches such as the Count-Min Sketch \cite{cormode2005countmin} 
\storyboard summaries include the values $x$.
We assume we have enough memory and compute to generate
$S$, making our routines closer to coreset construction \cite{philips2017coreset}
than streaming sketches \cite{muthukrishnan2005data}.
More details on how the values $x$ and counts $\gamma$ are chosen for each summary
are given in Section~\ref{sec:coop_construction} for interval aggregations and
Section~\ref{sec:pps} for cube aggregations.

\subsection{Query Processing}
\label{sec:sys_query_proc}
After the summaries have been constructed,
the \storyboard query processor can return query estimates $\hat{g}_{\QQ}(x)$ for different
aggregations $\QQ$ by using the summaries $S_i$ as proxies for the segments $\DD_i$.
Then, using $g$ to denote a generic query function, 
we can derive frequency or rank estimates over a query aggregation $\QQ$ by
adding up the estimates for the segment summaries.
\begin{equation*}
f_{S}(x) \coloneqq \sum_{x_j \in S} \gamma_j \cdot 1_{x_j = x} \qquad
r_{S}(x) \coloneqq \sum_{x_j \in S} \gamma_j \cdot 1_{x_j \leq x}
\end{equation*}
\begin{equation}
\hat{g}_{\QQ}(x) = \sum_{S_i \in \QQ} g_{S_i}(x)
\label{eqn:approx_query}
\end{equation}
For single rank or frequency estimates $\hat{g}_{\QQ}(x)$
the query processor can precisely add up scalar estimates using
Equation~\ref{eqn:approx_query}.
This is more efficient than merging mergeable summaries \cite{agarwal2012mergeable},
since we are just accumulating scalars, and potentially more accurate
as we will discuss in Section~\ref{sec:sys_error}.

%% Accumulator
\storyboard uses a richer accumulator data structure $A$ to support quantile and top-k / heavy hitter queries.
When there is sufficient memory, $A$ tracks the proxy values and counts in $S_1,\dots,S_k$. 
We then sort the items in $A$ by value to estimate quantiles or by count to estimate heavy hitters.
When memory is constrained, we instead let $A$ be a standard but very 
large stream summary of the proxy values and counts
stored in $S_1,\dots,S_k$.
We specifically use a Space Saving sketch~\cite{metwally2005efficient} for heavy hitters 
and a \pps (\texttt{VarOpt}~\cite{cohen2011structure}) sample for quantiles.
Then we can query $A$ to get quantile or heavy hitters estimates.
In practice the space $s_A$ available to $A$ is orders of 
magnitude greater than the space $s$ 
available to any precomputed summary, 
i.e. 50,000$\times$ in the deployment described in Section~\ref{sec:sys_motivation}.

\subsection{Error Model}
\label{sec:sys_error}
Consider the absolute (i.e. unscaled) error $\ea_{\QQ}$ which is the difference
between the true and estimated item frequency counts, or the difference between the
true and estimated ranks for a query aggregation $\QQ$.
Throughout the paper, we will use absolute errors $\ea$ for analysis, 
when comparing final query quality we use the relative (scaled) error 
$\er = \ea/|\QQ|$ \cite{cormode2010frequent} where $|\QQ| = \sum_{\rho_i \in \QQ} 1$.

When accumulating scalar rank or frequency estimates directly using Equation~\ref{eqn:approx_query}
the error $\ea_{\QQ}(x)$ is just the
sum of the errors introduced by the segment summaries for $\QQ$:
\begin{align}
\ea_{\QQ}(x) 
 % &= g_{\QQ}(x) - \hat{g}_{\QQ}(x) \\
 = \sum_{\DD_i \in \QQ} \ea_{\DD_i}(x) 
 = \sum_{\DD_i \in \QQ} \left(g_{\DD_i}(x) - g_{S_i}(x)\right)
 \label{eqn:error_accum_def}
\end{align}
When using the accumulator $A$ a quantile or heavy hitter estimate will be based off the
distribution defined by the $\hat{g}_{S}(x)$, but $A$ introduces its own 
additional error $\ea^{\text{(A)}}_{S}$
in approximating the proxy values in the summaries $S$, 
yielding a total error of:
\begin{equation}
\ea^{\text{(A)}}_{\QQ}(x) \leq |\ea_{\QQ}(x)| + |\ea^{\text{(A)}}_{S}(x)|.
\label{eqn:rawerror}
\end{equation}
Furthermore we are interested in systems that provide error bounds over all
values of $x$, so we consider the worst case error 
$
% \begin{equation}
\ea^{\text{(A)}}_{\QQ} \coloneqq \max_x |\ea^{\text{(A)}}_{\QQ}(x)|
% \end{equation}
$.
A bound on the maximum error over all $x$ also bounds the error of any
quantile or heavy hitter frequency estimate derived from the raw estimates
$\hat{g}$.

% Comparison and interpretation
To analyze the error, consider an aggregation $\QQ$ accumulating $k$ segments,
each with total weight $n=|\DD|=\sum_{x_i \in \DD} \delta_i$ and represented using summaries of size $s$.
Also, suppose that the accumulator $A$ has size $s_A \gg s$.
Suppressing logarithmic factors, state of the art frequency and quantile summaries 
have absolute error $O(n/s)$ \cite{karnin2016optimalquant,cormode2010frequent,agarwal2012mergeable,philips2017coreset}.
Different summarization techniques yield different errors as the 
size of the aggregation $k$ grows. 
\storyboard can reduce relative error significantly for large $k$.
We summarize the error bounds in Table~\ref{tab:summaryerror}.
\begin{table}[]
\centering
\caption{Summary Error ignoring constants combining $k$ summaries. The summaries used by storyboard: \coopq, \coopf, and \pps, all have reduced errors as for large $k$. \label{tab:errors}}
\label{tab:summaryerror}
\footnotesize
\begin{tabular}{p{3.0cm}|l|l}
\toprule
Summary & Relative $\er_{\mathcal{Q}}$ &Tot. Space\\
\midrule
\coopf & ${\log{k_T}}/({sk}) + 1/s_A$ & $sk$\\
\coopq & $\sqrt{k_T}/({sk}) + 1/s_A$ & $sk$\\
\pps \cite{cohen2011structure} & $1/(s\sqrt{k}) + 1/s_A$ & $sk$\\
\midrule
% Naive Accumulation & $\frac{1}{s} + \frac{1}{s_A} $ & $sk$\\
Mergeable \cite{agarwal2012mergeable} & $1/s$ & $sk$\\
Uniform Sample & $1/\sqrt{sk} + 1/s_A$ & $sk$\\
Hierarchical \cite{basat2018interval,qadarji2013hierarchical} & 
	$\log{k}/(sk) + 1/s_A $ & $sk\log{k_T}$\\
\bottomrule
\end{tabular}
\vspace{-1.5em}
\end{table}

Using mergeable summaries \cite{agarwal2012mergeable} for both $S$ and 
the accumulator $A$ gives us maximum absolute error
$\ea^{\text{(merge)}}_{\QQ} \leq O(|\QQ|/s) \leq O(kn/s)$
and maximum relative error
\begin{align}
\er^{\text{(merge)}}_{\QQ} \leq O(1/s).
\label{eqn:mergesumm_error}
\end{align}

\storyboard's accumulator $A$ applied to standard summaries gives us in the worst case
$\ea^{\text{(A)naive}}_{\QQ} \leq \sum_{\DD_i \in \QQ}|O(n/s)| + O(nk/s_A)$
so
\begin{align}
\er^{\text{(A)naive}}_{\QQ} \leq O(1/s) + O(1/s_A)
\label{eqn:naivesumm_error}
\end{align}
where the $O(1/s_A)$ is negligible for $s_A \gg s$.

However, \storyboard is able to achieve lower query error by reducing the sum of
errors from summaries in Equation~\ref{eqn:error_accum_def}.
By using independent, unbiased, weighted random samples -- 
specifically \pps summaries in Section~\ref{sec:pps} -- 
sums of random errors centered around zero will concentrate to zero, 
and one can use Hoeffding's inequality to
show that with high probability and ignoring log terms
$
\sum_{\DD_i \in \QQ} \ea_{\DD_i}(x) \leq O(\sqrt{k}n/s)
$ so
\begin{align} 
\er^{\text{(A)PPS}}_{\QQ} &\leq O\left(\frac{1}{\sqrt{k}{s}}\right) + O(1/s_A).
\end{align}
This already is lower than the relative error for mergeable summaries in 
Equation~\ref{eqn:mergesumm_error} for $k \gg 1$ and $s_A \gg s$.

In practice, Cooperative summaries (Section~\ref{sec:timerange}) achieve even better error than \pps summaries for interval queries.
We can prove that cooperative quantile summaries satisfy
$\max_{x} |\ea_{\QQ}^{(A)\coopq}(x)| \leq O(n\sqrt{k_T}/s)$
\begin{equation}
\begin{aligned}
\er_{\QQ}^{(A)\coopq} &\leq O\left(\frac{\min(\sqrt{k_T},k)}{ks}\right) + O(1/s_A)
\end{aligned}
\label{eqn:coopq_errorbound1}
\end{equation}
with much better empirical performance over intervals than \pps, while
cooperative frequency summaries satisfy \\
$\max_{x} |\ea_{\QQ}^{(A)\coopf}(x)| \leq O(n\log{k_T}/s)$
\begin{equation}
\begin{aligned}
\er_{\QQ}^{(A)\coopf} &\leq O\left(\frac{\min(\log{k_T},k)}{ks}\right) + O(1/s_A)
\end{aligned}
\label{eqn:coopf_errorbound1}
\end{equation}
where $k_T$ is the maximum length of an interval.
See Section~\ref{sec:interval-error} for more details and proof sketches.

Hierarchical estimation is a common solution for interval (range) queries \cite{basat2018interval,cormode2005countmin} and show up in differential privacy as well \cite{cormode2019rangeprivacy}.
We will describe an instance of these methods to illustrate their error scaling.
A dyadic (base $2$) hierarchy stores summaries of size $s \cdot 2^h$ for 
$h = 1\dots \log{k_T}$ to track segments of different lengths.
They can thus estimate intervals of length $k$ with error $\ea_{\QQ}^{\text{(A)Hier}} \leq O(n\log{k}/s) + O(1/s_A)$, similar to our cooperative frequency sketches.
However, they incur an additional $\log{k_T}$ factor in space usage to maintain
their multiple levels of summaries and provide worse error
empirically than our Cooperative summaries.
\section{Cooperative Summaries}
\label{sec:timerange}
In this section we describe the cooperative summarization algorithms \storyboard uses
for interval queries.
These summaries achieve high query accuracy when combined 
by compensating for accumulated errors over sequences of summaries.

\subsection{Interval Summary Construction}
\label{sec:coop_construction}
For interval queries, we assume that users specify a space limit $s$ to use
for summarizing an incoming data segment $\DD$.
Cooperative summaries then must make efficient use of their $s$ samples ($x_i \mapsto \gamma_i$)
to accurately represent a local segment of data.
To match state of the art summaries, we want
\begin{equation}
\max_{x}|\hat{g}_{S}(x) - g_{\DD}(x)| \leq r|\DD|/s
\label{eqn:localerror_def}
\end{equation} 
for an accuracy parameter $r\geq 1$.
However, there are many possible ways to choose the items to store
in $S$ that would satisfy Equation~\ref{eqn:localerror_def}. 

Within these constraints, \storyboard can choose $x_i,\gamma_i$
to minimize the total error for queries that aggregate multiple summaries.
\storyboard explicitly minimizes the error over a set of queries with 
fixed start points every $k_T$ segments.
We call these aggregation intervals ``prefix'' intervals $\pre_t$, a modification of
standard prefix-sum ranges \cite{Ho1997RangeQueries}.
\begin{equation}
\pre_t = \{\DD_{k_T \lfloor t/k_T\rfloor}, \dots,\DD_{k_T\lfloor t/k_T\rfloor+t\bmod k_T}\}.
\end{equation}
\begin{figure}
\centering
\includegraphics[width=.8\columnwidth]{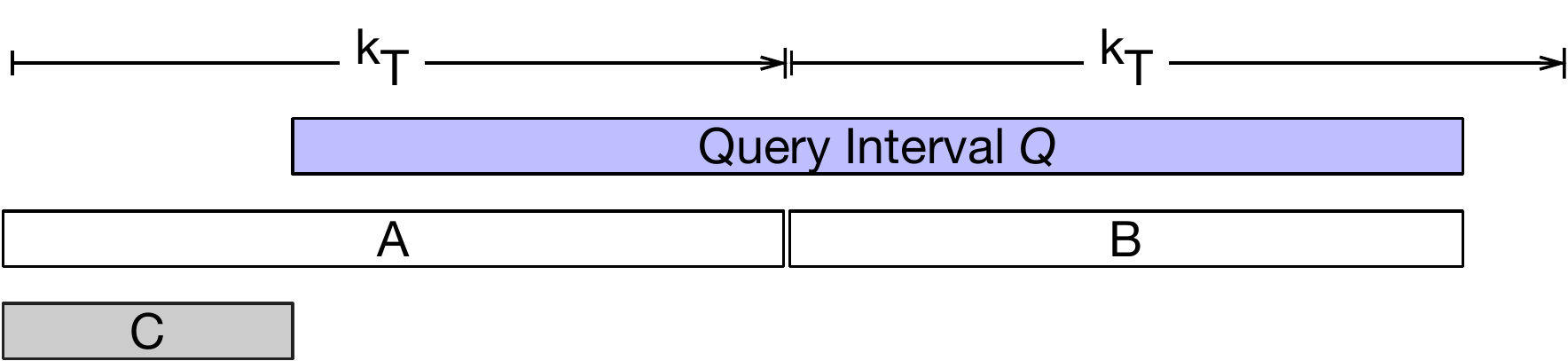}
\caption{Any contiguous interval can be expressed as a linear combination of aligned intervals $\pre_t$. 
In this example, $Q$ is expressed as $A \cup B \setminus C$}
\label{fig:ranges}
\end{figure}
Figure~\ref{fig:ranges} illustrates how any consecutive interval of up to $k_T$ segments
can be represented as an additive combination of up to $3$ prefix intervals.
As long as prefix intervals have bounded error $\ea_{\pre_t} = g_{\pre_t} - \hat{g}_{\pre_t}$, 
any contiguous interval up to length $k_T$ has error at most $3\ea$.
In order to minimize $\ea_{\pre_t}(x)$ \storyboard will have to track
its exact values, which may be resource intensive but is done during
data ingest.
The details of the summary construction algorithm differ for frequencies
and ranks.

\begin{algorithm}
\begin{algorithmic}
\Function{CoopFreq}{$\DD_t,s$}
	\State $h \gets |\DD_t|/s$
	\State $\ea_{\pre_t}(x) \gets \ea_{\pre_{t-1}}(x) + f_{\DD_t}(x)$
	\State $S_t \gets \{x\mapsto f_{\DD_t}(x) \ssep f_{\DD_t}(x) \geq h\}$ 	\Comment Heavy hitters
	\While {$|S_t| < s$} \Comment Correct Accumulated Errors
		\State $x_{m} \gets \arg \max_{x \in \pre_t \setminus S_t} \left(\ea_{\pre_{t}}(x)\right)$
		\State $\delta_{m} \gets \min\left(r\cdot h, \ea_{\pre_{t}}(x)\right)$
		\State $S_t \gets S_t \cup \{x_{m} \mapsto \delta_m\}$
		\State $\ea_{\pre_t}(x) \gets \ea_{\pre_{t}}(x) - \delta_m \cdot 1_{x = x_s}$
	\EndWhile
	\State \textbf{return} $S_t$
\EndFunction
\end{algorithmic}
\caption{Cooperative Item Frequencies Summary}
\label{alg:greedy_freq}
\end{algorithm}
In Algorithm~\ref{alg:greedy_freq} we present pseudocode for constructing a cooperative
summary of size $s$ for frequency estimates on a data segment $\DD_t$.
To satisfy Equation~\ref{eqn:localerror_def} and accurately represent $\DD_t$,
we store the true count for any segment-local 
heavy hitter items in $\DD_t$ that occur with count greater than $|\DD_t|/s$.
The remaining space in the summary is allocated to compensating the $x$ with the highest
cumulative undercount $\ea_{\pre_{t}}(x)$ thus far so that
overcounting $x$ in $S_t$ will adjust for
the undercount in the other summaries going forward.
For each of these compensating $x$, we store the smaller of
$r|\DD|/s$ and $\ea_{\pre_{t}}(x)$. This ensures Equation~\ref{eqn:localerror_def} is satisfied
and also keeps $\ea_{\pre_t}(x)$ positive, a useful invariant for proofs later.
Larger $r$ allow the algorithm trade off higher local error for less error accumulation across summaries.

\begin{algorithm}
\begin{algorithmic}
% \Function{CoopQuant}{$\DD_t,s$}
% 	\State $h \gets |\DD_t|/s$
% 	\State $S_t \gets \{\}$; $\ea_{\pre_t} \gets \ea_{\pre_{t-1}}$

% 	\State $\DD_{t1},\dots,\DD_{ts} \gets \Call{Partition}{\DD_t,s}$ \Comment Sorted Chunks
% 	\For{$i \in 1\dots s$}
% 		\State $L(z) \coloneqq \sum_{y \in U} \pot(\ea_{\pre_{t-1}}(y)-1_{y \geq z}h)$
% 		\State $x_s \gets \arg \min_{z \in \DD_{ti}} L(z)$ \Comment Minimize Loss
% 		\State $S_t \gets S_t \cup \{x_s \mapsto h\}$
% 	\EndFor
% 	\State \textbf{return} $S_t$
% \EndFunction
\Function{CoopQuant}{$\DD_t,s$}
	\State $h \gets |\DD_t|/s$; $S_t \gets \{\}$
	\State $\ea_{\pre_t}(x) \gets \ea_{\pre_{t-1}}(x) + r_{\DD_t}(x)$

	\State $\DD_{t1},\dots,\DD_{ts} \gets \Call{Partition}{\DD_t,s}$ \Comment Sorted Chunks
	\For{$i \in 1\dots s$}
		\State $L(z) \coloneqq \sum_{y \in U} \pot(\ea_{\pre_{t}}(y))$
		\State $x_s \gets \arg \min_{z \in \DD_{ti}} L(z)$ \Comment Minimize Loss
		\State $S_t \gets S_t \cup \{x_s \mapsto h\}$
		\State $\ea_{\pre_t}(x) \gets \ea_{\pre_{t}}(x) - h\cdot 1_{x \geq x_s}$
	\EndFor
	\State \textbf{return} $S_t$
\EndFunction
\end{algorithmic}
\caption{Cooperative Quantile Summary}
\label{alg:greedy_quant}
\end{algorithm}
In Algorithm~\ref{alg:greedy_quant} we present pseudocode for constructing a cooperative
summary of size $s$ for rank estimates on a data segment $\DD_t$.
To satisfy Equation~\ref{eqn:localerror_def} and accurately represent $\DD_t$, 
we sort the values in $\DD$ and partition the sorted values into $s$ equally sized chunks.
Then \coopq selects one value in each chunk to include in $S_t$ as a representative
with proxy count $|\DD|/s$.
This ensures that the any rank can be estimated using $S_t$ with error 
at most $|\DD|/s$.
Within each chunk, we store the item that minimizes a total loss
$L = \sum_{x \in U} \pot(\ea_{\pre_t}(x))$ with
$\pot(\epsilon)=\cosh{(\alpha\epsilon)}$,
$\alpha=s/({\sqrt{k_T}n_{\text{max}}})$,
$n_{\text{max}} = \max_{t} |\DD_t|$ the maximum size of a data segment,
and $k_T$ the maximum interval length.
$\cosh(x) = \frac{1}{2}(e^{x} + e^{-x})$ is used in discrepancy theory \cite{spencer1977balancing} to exponentially penalize
both large positive and large negative errors, so $L$
serves as a proxy for the $L_\infty$ maximum error.
Note that we need to bound $n_\text{max}$
to set $\alpha$ for this algorithm, though in practice accuracy
changes very little depending on $n_\text{max}$
\subsection{Interval Query Error}
\label{sec:interval-error}
\coopf and \coopq both provide estimates with local error $\ea_{\DD}(x) \leq r |\DD|/s$
for a single segment $\DD$, and minimize the cumulative error over $\ea_{\pre_t}(x)$
prefix intervals (and thus general intervals).
In this section we analyze how $\ea_{\pre_t}(x)$ grows with $t$.
This allows us to prove Equations \ref{eqn:coopq_errorbound1} and \ref{eqn:coopf_errorbound1} in Section~\ref{sec:sys_error} which
bound the query error from accumulating results over any sequence of $k_T$
summaries.

The general strategy will be to define a loss $L_t$ which is a function of the
errors $\ea_{\pre_t}(x)$ parameterized by a cost function $\pot$
\begin{equation}
L_t \coloneqq \sum_{x \in U} \pot\left(\ea_{\pre_t}(x)\right)
\end{equation}
where $U$ is the universe of observed values $x \in |\pre_t|$.
We can bound the growth of $L_t$ when \coopq and \coopf are used
to construct sequences of summaries.
Then, we can relate $L_t$ and $\max_x |\ea_{\pre_t}(x)|$ to bound the latter.
Omitted proofs in this section can be found in Appendix \ref{sec:freq_pot_proof}.

\subsubsection{\coopf Error}
For frequency summaries, we use the cost function
$\pot(x) = \exp(\alpha x)$ for a parameter $\alpha$.
We minimize a sum of exponentials as a proxy for the maximum error.
Lemma~\ref{lem:coopfreqpot} bounds how much $L_t$ can increase with $t$.
\begin{lemma}
When \coopf constructs a summary with size $s$ for $\DD_t$ the loss satisfies 
$$L_t \leq L_{t-1} + \alpha r |\DD_t|$$
for $\pot(x) = \exp(\alpha x)$ as long as
$0 < \alpha \leq 2\frac{s}{|\DD_{t}|}\frac{r-1}{r^2}$.
\label{lem:coopfreqpot}
\end{lemma}
Given this, we can bound the cumulative error:
\begin{theorem}
\coopf maintains
$$
\max_{x \in U} |\ea_{\pre_t}(x)| \leq \frac{1}{\alpha} \ln\left(1+\alpha r \sum_{i=1}^{t} |\DD_i|\right)
$$ 
where $\alpha = 2\frac{s}{\max_i{|\DD_i|}}\frac{r-1}{r^2}$.
\label{thm:coopf}
\end{theorem}
\begin{proof}
$\alpha$ satisfies the conditions in Lemma~\ref{lem:coopfreqpot} so
$L_t \leq L_0 + \alpha r \sum_{i=0}^t |\DD_i|$ and thus
\begin{align*}
\exp(\alpha \max_{x \in U} \ea_t(x))-1 \leq \alpha r \sum_{i=0}^t |\DD_i|
\end{align*}
\end{proof}
To illustrate the asymptotic behavior we can apply Theorem~\ref{thm:coopf} 
with $r=\frac{3}{2}$ and consistent segment weights $n=|\DD_i|$ to see 
in Corrollarry~\ref{cor:freq_concrete} that the absolute error grows logarithmically with the number of segments $k$ in the interval.
\begin{cor}
For $r=\frac{3}{2}$ and $|\DD_i|=n$, \coopf maintains
$$
\max_{x \in U} |\ea_k(x)| \leq \frac{9}{4}\frac{n}{s} \ln\left(1+\frac{2}{3}nk\right)
$$ 
\label{cor:freq_concrete}
\end{cor}
In fact this result is close to optimal: an adversary generating incoming data can guarantee at least $\Omega(\log{k})$ error accumulation by generating data containing items the summaries have undercounted the most so far.

\subsubsection{\coopq Error}
For rank queries we use the cost function $\pot(x) = \cosh{\alpha x}$.
Since $\cosh{z} = \frac{1}{2}(\exp(z) + \exp(-z))$ this exponentially 
penalizes both under and over-estimates symmetrically, and is thus a smooth
proxy for the maximum absolute error of the error, also used in discrepancy
theory \cite{spencer1977balancing}.
As with \coopf, Lemma~\ref{lem:quant_cosh} bounds how much $L_t$ can increase with $t$.
\begin{lemma}
When \coopq constructs a summary with size $s$ for $\DD_t$ the loss function satisfies
$$
L_t \leq L_{t-1} \exp{\alpha^2(|\DD_t|/s)^2/2}
$$
for $\pot(x) = \cosh(\alpha x)$
\label{lem:quant_cosh}
\end{lemma}
From Lemma~\ref{lem:quant_cosh} we can bound the maximum rank error.
\begin{theorem}
\coopq maintains
$$
\max_{x \in U} |\ea_{\pre_t}(x)| \leq \frac{1+ 2\ln{(2|U|)}}{2s}\sqrt{\sum_{i}^{t} |\DD|_t^2}
$$
with $\pot(x) = \cosh(\alpha x)$ and
$\alpha = s\left(\sum_{i=0}^{t} |\DD_i|^2\right)^{-1/2}$.
\end{theorem}
\begin{proof}
Using Lemma~\ref{lem:quant_cosh}:
\begin{align*}
L_t &\leq L_0 \exp{\left(\alpha^2/2\sum_{i=0}^{t}\left(\frac{|\DD_i|}{s}\right)^2\right)} \\
% \frac{1}{2}\exp{(\alpha \max_{x \in U} |\ea_t(x)|)} &\leq |U| \exp{\left(\alpha^2/2\sum_{i=0}^{t}\left(\frac{|\DD_i|}{s}\right)^2\right)} \\
\max_{x \in U} |\ea_t(x)| &\leq \frac{1}{\alpha}\ln(2|U|) + \frac{\alpha}{2} \sum_{i=0}^{t}\left(\frac{|\DD_i|}{s}\right)^2
\end{align*}
Then setting $\alpha = s/{\sqrt{\sum_{i=0}^{t} |\DD_i|^2}}$
completes the proof.
\end{proof}
This can be instantiated for data segments with constant total weight in Corrollary~\ref{cor:quant_concrete_imp}, which shows that \coopq has error $O(\sqrt{k}/s)$.
\begin{cor}
For $|\DD_t|=n$ constant and $\pot(x) = \cosh(\alpha x)$ 
with $\alpha = \frac{s}{n\sqrt{k}}$,
\coopq maintains
$$
\max_{i \in U} |\ea_{\pre_k}(i)| \leq \frac{n}{2s}\left(\sqrt{k}+2\ln{(2|U|)}\right)
$$ 
\label{cor:quant_concrete_imp}
\end{cor}

\section{Optimizing Cube Queries}
\label{sec:optcube}
In this section we describe the weighted probability proportional to
size samples (\pps) \cite{cohen2011structure,ting2018disagg} \storyboard uses to
summarize segments for both frequency and rank data cube aggregations. 
These randomized summaries have errors which cancel out with high probability.
Then we describe how \storyboard optimizes the allocation of space and bias to these summaries
to increase average query accuracy further for target cube workloads.

\subsection{\pps Summaries}
\label{sec:pps}
A \pps summary is a weighted random sample that 
includes items with probability proportional to their size or total count
in a data segment $\DD$
\cite{cohen2011structure,ting2018disagg,huang2011optimalsamp}.
Values $x_i$ with true occurence count $\DD(x_i) = \delta_i$ are sampled for inclusion
in the summary $S$ according to Equation~\ref{eqn:ppsweight}
\begin{align}
\label{eqn:ppsweight}
\Pr[x_i \in S] &= \min(1, \delta_i/h) \\
S(x_i) &= \begin{cases}
h & \delta_i \leq h\\
\delta_i  & \delta_i > h
\end{cases}.
\end{align}
For an accuracy parameter $h$, 
heavy hitters that occur more than $h$ times are always sampled with their true
count, while those with count $0 \leq \delta_i \leq h$ are either included with a 
proxy weight of $h$ or excluded from the summary. 
Thus, $S(x_i)$ is an unbiased estimate for $\delta_i$ with maximum local error of $h$.
We will see later that rank estimates 
$\hat{r}_{S}(x) \sum_{x_j \in S} \gamma_j \cdot 1_{x_j \leq x}$
can also bounded by $h$.

Setting $h \leq |\DD|/s$ ensures that the summary will have expected size $s$,
but is conservative.
In Algorithm~\ref{alg:calct} we present the procedure from
\cite{cohen2011structure} we use to set $h$ to minimize error 
while keeping the summary size at most $s$
by excluding the effect of heavy hitters.
\begin{algorithm}
% \small
\begin{algorithmic}
\Function{CalcT}{$\DD,s$}
\State $h \gets |\DD|/s$
\State $H \gets \{\}$ \Comment{Local Heavy Hitters}
\While{$\max_{x \in \DD \setminus H} f_{\DD}(x) \geq h$}
\State $x_{\text{max}} \gets \arg\max_{x \in \DD \setminus H} \DD(x)$
\State $H \gets H \cup \{x_{\text{max}}\}$ 
\State $h \gets \frac{\sum_{x \in \DD \setminus H} f_{\DD}(x)}{s-|H|}$
\EndWhile
\State \Return $h$
\EndFunction
\end{algorithmic}
\caption{Calculate minimal $h$ threshold}
\label{alg:calct}
\end{algorithm}

One way to implement \pps is to independently sample items according to Equation~\ref{eqn:ppsweight}, 
but this does not guarantee the summary will store exactly $s$ values.
Instead we use the \texttt{PairAgg} procedure in Algorithm~\ref{alg:pairagg} to transform
sampling probabilities for pairs of items
until we have $s$ or $s-1$ values with probability $1$. 
We can do so in a way that guarantees that the error $\max_x |\ea(x)| \leq h$ and is unbiased with $E[\ea(x)] = 0$ for both frequency and rank queries. See \cite{cohen2011structure} for details.
\begin{algorithm}
\small
\begin{algorithmic}
\Function{PairAgg}{$p_i,p_j$}
\If{$p_i + p_j < 1$}
	\If{$\text{rand}() < p_i/(p_i+p_j)$} $p_i \gets p_i + p_j$; $p_j \gets 0$ 
	\Else\, $p_j \gets p_i + p_j$; $p_i \gets 0$ \EndIf
\Else
	\If{$\text{rand}() < \frac{1-p_j}{2-p_i-p_j}$} $p_i \gets 1$; $p_j \gets p_i+p_j-1$ 
	\Else\, $p_i \gets p_i + p_j-1$; $p_j \gets 1$ \EndIf
\EndIf
\EndFunction
\end{algorithmic}
\caption{Pair Aggregation for \pps}
\label{alg:pairagg}
\end{algorithm}
% See \cite{cohen2011structure} for details.

% Given inclusion probabilities $p_1,\dots,p_m$ defined by Equation~\ref{eqn:ppsweight}, \pps executes \texttt{PairAgg} on pairs of items until all probabilities are either $0$ or $1$ and at most one remaining element is sampled with probability $0 < p_j < 1$, yielding a summary with exactly $s-1$ or $s$ stored values.
% \pps can be used to construct summaries for either frequencies or quantiles depending on the
% order in which \texttt{PairAgg} is executed.
% For frequency queries, any execution order is valid, while for
% for quantile queries we first sort the items $x_i \in \DD$ in ascending order $x_1,\dots,x_n$
% and, following the \texttt{VarOpt} method in \cite{cohen2011structure}, use \texttt{PairAgg} to merge sequential pairs of items at a time.
% This guarantees that $\forall x.|r_{\DD}(x) - \hat{r}_{S}(x)| \leq h$.
% \begin{algorithm}
% \begin{algorithmic}
% \Function{PairAgg}{$p_i,p_j$}
% \If{$p_i + p_j < 1$}
% 	\If{$\text{rand}() < p_i/(p_i+p_j)$} $p_i \gets p_i + p_j$; $p_j \gets 0$ 
% 	\Else\, $p_j \gets p_i + p_j$; $p_i \gets 0$ \EndIf
% \Else
% 	\If{$\text{rand}() < \frac{1-p_j}{2-p_i-p_j}$} $p_i \gets 1$; $p_j \gets p_i+p_j-1$ 
% 	\Else\, $p_i \gets p_i + p_j-1$; $p_j \gets 1$ \EndIf
% \EndIf
% \EndFunction
% \end{algorithmic}
% \caption{Pair Aggregation for \pps}
% \label{alg:pairagg}
% \end{algorithm}

\subsection{Cube Summary Construction}
\label{sec:cubeworkload}
For data cubes,
storyboard maintains a collection of \pps summaries for data segments $\DD_i$
that form a complete, atomic partition of combinations of dimension values.
A cube query then specifies a set of segments $\QQ_s = \{\DD_1, \dots, \DD_k\}$
that match dimension value filters.
\storyboard ingests a complete cube dataset in batch and is given
a total space budget $S_T$ for storing summaries.
This opens up the opportunity for optimizations across the entire
collection of summaries.

In most multi-dimensional data cubes some queries and dimension values will be much rarer than others.
This makes it wasteful to optimize for worst-case error: even the rarest
data segment would require the same error and space as more representative
segments of the cube.
Thus, in \storyboard we make use of limited space by optimizing for
the average error of queries sampled from a probabilistic workload $W$ specified by
the user.
We show in Section~\ref{sec:cube_lesion} that the workload does not have to be perfectly
specified to achieve accuracy improvements.

We consider a workload $W$ as a distribution 
over possible queries $\QQ_i$ where $\Pr[\QQ_i \sim W] = q_i$.
This is based off of the workloads in STRAT~\cite{chaudhuri2007strat}, 
though STRAT targets only count and sum queries using simple uniform samples.
To limit worst-case accuracy, we can optionally impose a minimium size for each 
segment summary $s_{\text{min}}$ so that the maximum relative error for any query 
is $\er \leq \frac{1}{s_{\text{min}}}$.

\subsection{Minimizing Average Error}
Consider the error incurred by combining summaries over a query $\QQ = \DD_1,\dots,\DD_k$, where the segment summaries $S_i$ have size $s_i$ and represent segments with total count
$|\DD_i| = n_i$. 
Then, based on Equation~\ref{eqn:ppsweight}, the relative error $\er_{\QQ}(x)$ is a random variable that depends on the items selected for inclusion in the \pps summaries.
We will bound the mean squared relative error $E[\er(x)^2]$.

For a single segment $\DD_i$, the \pps summary is unbiased and returns both frequency and rank estimates that lie within a possible range of length $h$. 
Thus, the absolute error satisfies $E[\ea_{\DD_i}(x)] = 0$ and 
$E[\ea_{\DD_i}(x)^2] \leq \frac{1}{4}h^2 \leq \frac{1}{4}n_i^2/s_i^2$, and
since the summaries $S_i$ are independent:
$$
E[\ea_{\QQ}^2] \leq \frac{1}{4}\sum_{\DD_i\in \QQ}\left(\frac{n_i}{s_i}\right)^2.
$$
\minihead{Space Allocation}
Now, we minimize the mean squared relative error (MSRE) for queries drawn from a workload $\QQ_z \sim W$ where $\Pr[\QQ_z] = q_z$. Let $|\QQ_z| = \sum_{\DD_i \in \QQ_{z}}|\DD_i|$.
\begin{align}
E_{\QQ_{z} \sim W}\left[\er_{\QQ}^2\right] 
% &\leq\frac{1}{4} E_{\QQ_{z}\sim W}\left[\sum_{\DD_i \in \QQ_{j}}\left(\frac{n_{i}}{s_{i}}\right)^{2}|\QQ_z|^{-2}\right]\\
&\leq \frac{1}{4}\sum_{\DD_i\in \mathcal{\DD}}\frac{n_i^2}{s_i^2}\left(\sum_{z|\DD_{i}\in \QQ_{z}}q_z|\QQ_z|^{-2}\right)
\label{eqn:bigmsebound}
\end{align}
We can solve for the $s_i$ that minimize the RHS of Equation~\ref{eqn:bigmsebound}
under the total space constraint that $\sum_i s_i = S_T$
using Lagrange multipliers.
The optimal $s_i$ are $s_i \propto \alpha_i^{1/3}$ where
\begin{equation}
\alpha_{i}=n_{i}^{2}\sum_{z|\DD_{i}\in \QQ_{z}}q_{z}|\QQ_z|^{-2}
\label{eqn:alphaeq}
\end{equation}
Since we can compute $\alpha_i$ given $W$, this gives us a closed form expression for an allocation of storage space.

\minihead{Bias and Variance}
When estimating item frequencies, we can further reduce error by tuning the bias of \pps summaries to reduce their variance.
Though this does not generalize to quantile queries, the improvements in accuracy for frequency queries can be substantial,
and we have not seen other systems optimize for bias across a collection of summaries.

For example, consider a segment $\DD$ with $n>4$ unique items that each only occur once.
If we summarize the data with an empty summary, estimating $0$ for the count of each item,
we introduce a fixed bias of $1$ but have a deterministic estimator with no variance.
This substantially reduces the error compared to an unbiased \pps estimator constructed on $\DD$ which will have variance $n^2(\frac{1}{n}\cdot (1-\frac{1}{n})) = n(1-\frac{1}{n}) > 3$.

In general, if we have a segment $\DD$ consisting of item weights $\{x_i \mapsto \delta_i\}$
then we bias the frequency estimates 
$\hat{f}_{\DD}(x)$
by subtracting $b$ from the count of every distinct element in $\DD$ before constructing a \pps summary,
and then adding $b$ back to the stored weights.
During \pps construction, $h$ and thus the variance is reduced because $\DD$ has a lower effective total weight $n_i[b]$ given by
\begin{equation}
n_{i}\left[b\right]=\sum_{x_i\in \DD}\left(\delta_i-b\right)^{+}
\label{eqn:biasedn}
\end{equation}
where $(x)^{+}$ is the positive part function $(x)^{+}=\text{max}(x,0)$.

The error for a single segment $\DD_i$ is now bounded by
$\ea_i \leq b_i + \nu_i$ where $b_i$ is the bias and $\nu_i$ is the remaining unbiased \pps error on the bias-adjusted weights, so the
MSRE for a query $\QQ$ is:
% $E\left[\ea_{\QQ}^{2}\right] \leq E\left[\left(\sum_{\DD_i\in \QQ} b_i + \nu_i\right)^2\right]$
% The $\nu_i$ are \pps errors so they are independent with $E[\nu_i] = 0$, and $E[\nu_i^2] \leq \frac{1}{4}\frac{(n_i[b_i])^2}{s_i^2}$.
\begin{align}
E\left[\er_{\QQ}^{2}\right] & \leq|\QQ|^{-2}\left(\left(\sum_{\DD_i\in \QQ}b_{i}\right)^{2}+\sum_{\DD_i\in \QQ}\frac{1}{4}\left(\frac{n_{i}\left[b_{i}\right]^{2}}{s_{i}^{2}}\right)\right) 
 \label{eqn:biaserror}
\end{align}
Equation~\ref{eqn:biasedn} shows that $n[b]$ is convex with respect to $b$ since it is a sum of convex functions ($\text{max}$ is convex), so the
RHS of Equation~\ref{eqn:biaserror} is convex as well.

% \pagebreak
\minihead{Recap}
In summary \storyboard does the following for cube aggregations.
\begin{enumerate}
	\item Set summary sizes $s_i \propto \alpha_i^{1/3}$ using Equation~\ref{eqn:alphaeq}, scaled so $\sum s_i = S_T$. 
	\item Solve for biases $\vec{b}$ that minimize the RHS of Equation~\ref{eqn:biaserror} for $\QQ$ a query over the entire dataset.
	\item Construct \pps summaries according to $\vec{s}$ and $\vec{b}$
\end{enumerate}
We optimize Equation~\ref{eqn:biaserror} using the LBFGS-B solver \cite{Byrd1995lbfgs} in SciPy \cite{scipy}.
To simplify computation we optimize $b_i$ for a single aggregation: the whole cube.
An optimal setting of $\vec{b}$ for this whole cube query will not increase relative error over any other query compared to $\vec{b}=0$.
Finding efficient and accurate proxies to optimize 
is a direction for future work.
\section{Evaluation}
\label{sec:eval}
In our evalution, we show that:
\begin{enumerate}
    \setlength\itemsep{.5em}
    \setlength{\topsep}{0em} 
	\item \storyboard's cooperative summaries achieve lower error as interval length increases compared with other summarization techniques: up to $8\times$ for frequencies and $25\times$ for quantiles (Section~\ref{sec:eval-interval}).
    \item \storyboard's space and bias optimizers provide lower average error for cube queries compared with alternative techniques, with reductions between $15$\% to $4.4\times$ (Section~\ref{sec:eval-cube}).
    \item \storyboard's accuracy generalizes across different system and summary parameters, including accumulator size, maximum interval length, and workload specification (Sections \ref{sec:eval_sys_parameters} and \ref{sec:eval_summ_design}).
\end{enumerate}

\subsection{Experimental Setup}
\label{sec:eval-setup}
\minihead{Error Measurement}
Recall from Section~\ref{sec:sys_error} that we are interested in error bounds that
are independent of a specific item or value $x$, so we look at the maximum
error over values $x$,
$\er_{\QQ} \coloneqq \max_x |\er_{\QQ}(x)|$. 
For a large domain of values $U$ it is infeasible to compute $\max_{x\in U}$
so for frequency queries we estimate this maximum over a sample of 200 items drawn
at random from the data excluding duplicates and
for rank queries over a set of 200 equally spaced
values from the global value distribution.
Following common practice for approximate summaries \cite{cormode2010frequent}, 
we rescale the absolute error by the total
size of the queried data to report relative errors $\er_{\QQ} = \ea_{\QQ} / |\QQ|$.

The final query error $\er_{\QQ}^{\text{(A)}}$ is then bounded by the sum of the
summary and accumulator errors $\er_{\QQ} + \er^{\text{(A)}}$.
In our evaluations we assume that the accumulator is large enough to introduce negligible
additional error,
and confirm that the additional error rapidly vanishes as the size of the
accumulator grows in Figure~\ref{fig:linear_acc}.
When summaries provide a native merge routine, we still benchmark query accuracy
by accumulating their estimates rather than merging the summaries directly.
This is strictly more accurate than merging, which will further compress intermediate results
to fit in the original summary space, and provides a more fair comparison with \storyboard.

\minihead{Implementation}
We evaluate a prototype implementation of \storyboard written in Python 
with core summarization and query processing logic compiled to C and
code available\footnote{\scriptsize\url{https://github.com/stanford-futuredata/sketchstore}}.
Since our focus is query accuracy under space constraints, 
our prototype is a single node in-memory system though it can be extended to a distributed
system in the same manner as Druid.

Our implementation of \coopf (Algorithm~\ref{alg:greedy_freq}) uses $r=1$ and sets
$h$ using \texttt{CalcT} in Algorithm~\ref{alg:calct} rather than letting $h = |\DD_t|/s$.
$h\coloneqq\texttt{CalcT}$ gives us better segment accuracy 
and the error bounds still hold under a modified proof.
We implement \coopq (Algorithm~\ref{alg:greedy_quant}) with a cost function parameter $\alpha$
set based on a maximum interval length of $k_T=1024$, and loss $L$ calculated over
the universe of elements seen so far when the full universe is not known ahead of time.

\minihead{Datasets}
We evaluate frequency estimates on
10 million destination ip addresses (\caida) from a Chicago Equinix backbone on 2016-01-21
available from CAIDA \cite{caida}, 
10 million items (\zipf) drawn from a Zipf (Pareto) distribution with parameter $s=1.1$, 
and 10 million records from a production service request log at \msft
with categorical item values for network service provider (\servicep) and OS Build (\serviceo).

We evaluate quantile estimates on 
2 million active power readings (\power) from the UCI Individual household electric power consumption
 dataset \cite{uci},
10 million random values (\uniform) drawn from a continuous uniform $U\sim [0,1]$ distribution,
and 10 million records from the same \msft request log with numeric traffic values (\servicet).

% ==============
% Overall Accuracy
% ==============
\begin{figure*}
    \centering
    \begin{subfigure}[t]{\textwidth}
    \includegraphics[width=\textwidth]{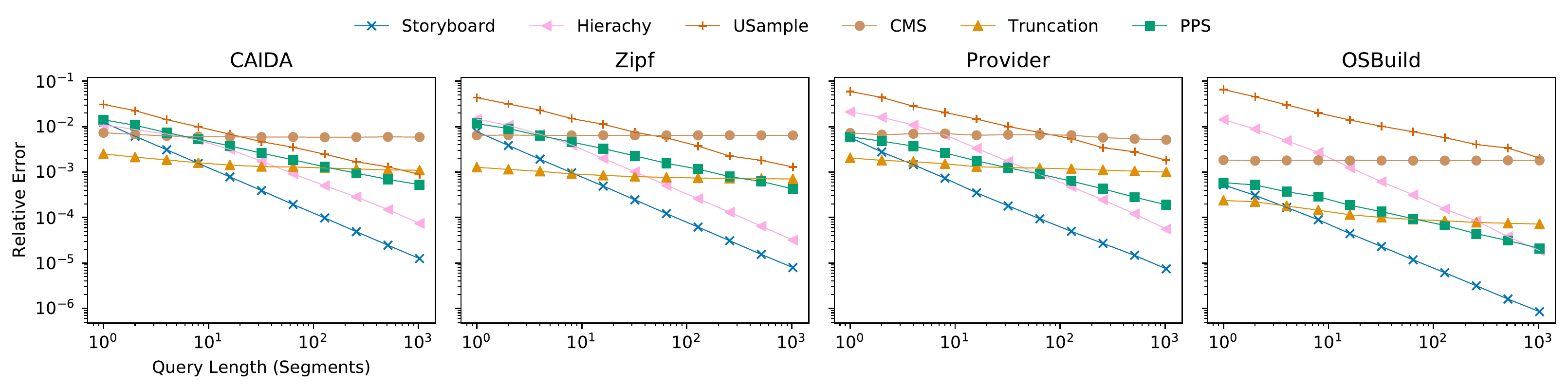}
    \vspace{-1em}
    \caption{Frequency Queries}\label{fig:linear_freq}
    \end{subfigure}
    \begin{subfigure}[t]{.75\textwidth}
    \includegraphics[width=\textwidth]{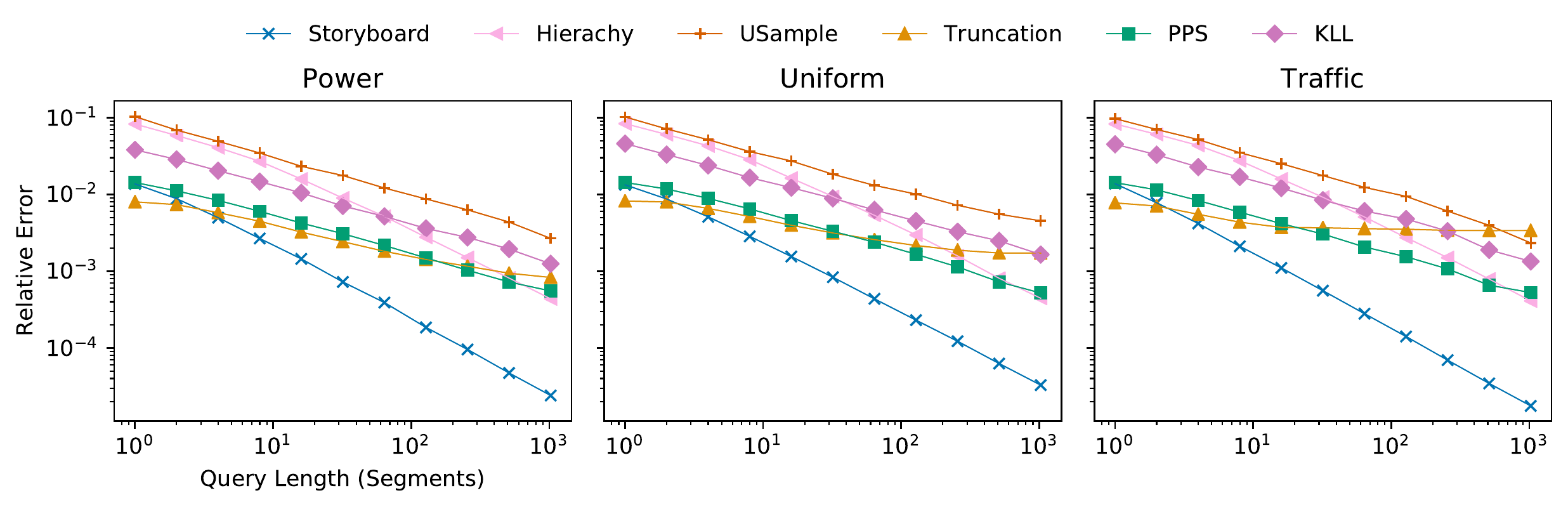}
    \vspace{-1em}
    \caption{Quantile Queries}\label{fig:linear_quant}
    \end{subfigure}
    \caption{Query error over interval queries of different lengths. \storyboard's cooperative summaries have increasingly high accuracy as the query length increases.}
\end{figure*}

\subsection{Overall Query Accuracy}
\minihead{Summarization Methods}
We compare a number of summarization techniques for frequencies and quantiles,
and configure them to match total space usage when comparing accuracy.
For all counter and sample-based summaries including \coopf, \coopq, and \pps, 
we set the number of counters or samples to the same $s$.

We compare against two popular mergeable summaries: 
the optimal quantiles sketch (\kll) from \cite{karnin2016optimalquant}, 
and the Count-Min frequency sketch (\cms) \cite{cormode2005countmin}.
For the count-min sketch we set $d=5$ and let the width $w=s$ parameter represent the space usage.
We also compare against uniform random sampling (\sampling) \cite{cormode2012synopses}, 
and optimal single-segment summaries (\truncation) 
that summarize a segment by storing the exact item counts for the top $s$ items, 
or storing $s$ equally spaced values for quantiles.

For interval queries we also compare with storing \truncation summaries in a hierarchy 
(\dyadic) following \cite{basat2018interval,cormode2019rangeprivacy}.
Specifically, the \dyadic summarization strategy
with base $b$ constructs $h$ layers of summaries.
Summaries in layer $i$ are allocated space $b^{i}\cdot s_0$ to summarize 
aligned intervals of $b^{i}$ segments.
Any query interval of length $k$ can be represented using $b\lceil{\log_{b}{k}}\rceil$ summaries
from different layers.
Since this requires maintaining $h=\log_{b}{k_T}$ layers, to fairly compare total space usage
we scale the space $s_0$ allocated to the lowest layer summaries
by a factor $s_0 = s / \log_{b}{k_T}$.
Unless otherwise stated we use $b=2$, though we will show in Section~\ref{sec:eval_coopdesign}
that the choice does not have a significant impact on accuracy.

For cube queries we also compare with cube AQP techniques that use uniform 
\sampling with different space allocations:
the \samplingp method uses \sampling summaries but allocates space proportional to each segment size
as a baseline random reservoir sample would \cite{vitter1985random}
while the \strat method uses the method in the STRAT AQP system \cite{chaudhuri2007strat}, 
which like \storyboard allocates space to minimize average error.

\subsubsection{Interval Queries}
\label{sec:eval-interval}
We first evaluate \storyboard accuracy on interval queries,
partitioning datasets with associated time or sequence columns 
into $k_T=2048$ size time segments.
Then, we construct summaries with storage size $s=64$.

In Figures \ref{fig:linear_freq} and \ref{fig:linear_quant} 
we show how relative query error $\er_{\QQ}$ varies
with the number of segments $k$ spanned by the interval.
For $k=1,2,4,\dots,1024$ we sample 100 random start and end times
for intervals with length $k$ and plot
the average and standard deviation of the query error.

\storyboard, which uses cooperative summaries, outperforms any system
that uses mergeable summaries or random sampling as $k$ increases.
As the interval length increases merging the 
mergeable summaries (\cms, \kll) maintain their error as expected.
Accumulating \truncation summaries also maintains the same constant error.
\dyadic, \pps, and \sampling are all able to reduce error when
combining multiple summaries, while Cooperative summaries outperform
all alternatives as $k$ exceeds 10 summaries.
We also observe that despite our weaker worst-case bounds for cooperative quantile
summaries, they achieve higher accuracy in practice compared to
alternative methods.
However, \storyboard gives up a constant factor in accuracy when
aggregating less than $10$ summaries compared to alternatives.

\subsubsection{Cube Queries}
\label{sec:eval-cube-setup}
\label{sec:eval-cube}

\begin{figure*}
    \centering
    \begin{subfigure}[t]{.65\textwidth}
    \includegraphics[width=\textwidth]{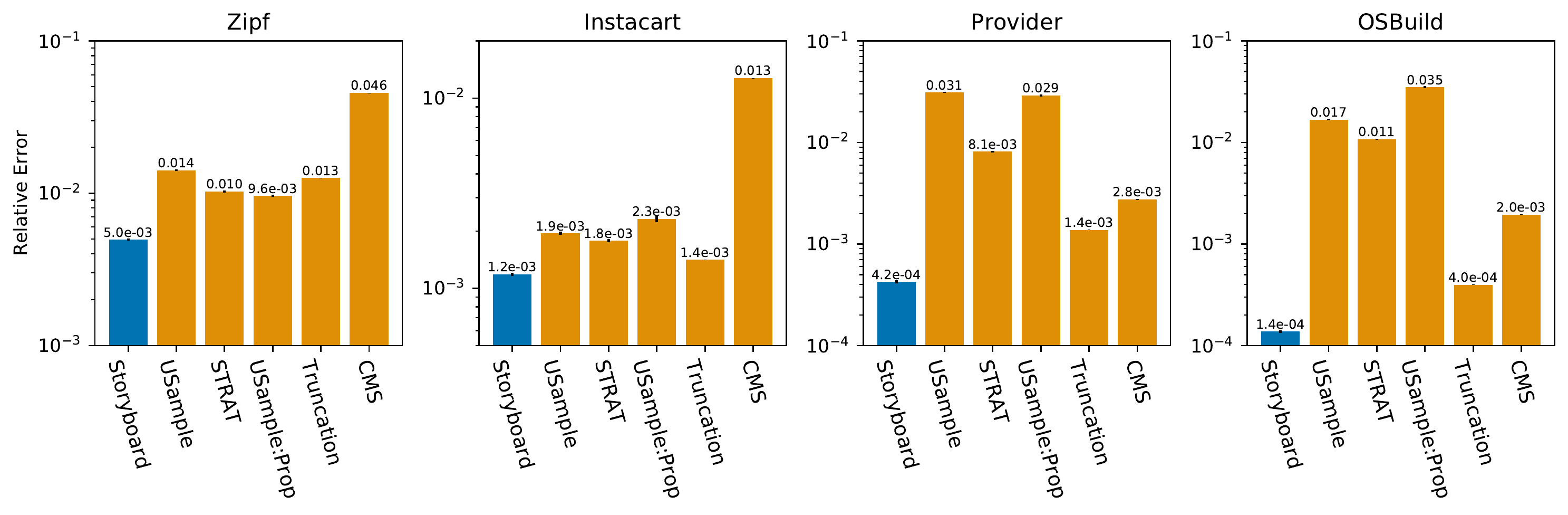}
    \caption{Frequency Queries}\label{fig:cube_freq_global}
    \end{subfigure}
    \begin{subfigure}[t]{.32\textwidth}
    \includegraphics[width=\textwidth]{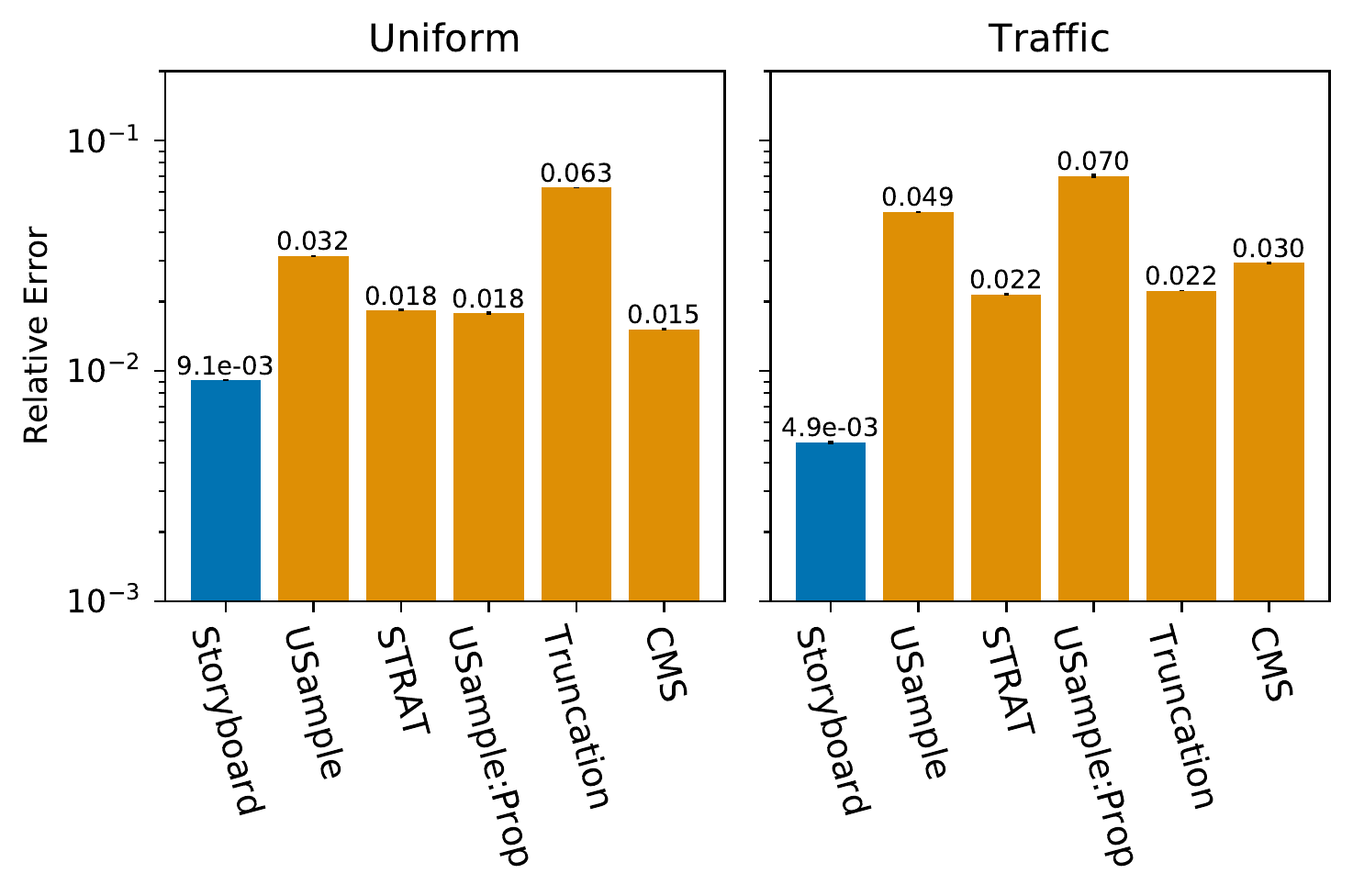}
    \caption{Quantile Queries}\label{fig:cube_quant_global}
    \end{subfigure}
    \caption{Average query error over a workload of cube queries.}
\end{figure*}

\begin{table}
\iftoggle{arxiv}{\small}{\scriptsize}
\centering
\caption{Cube Datasets}
\label{tab:cube_setup}
\begin{tabular}{l|l l l}
\toprule
Data & \# Segments & Summary Space\\
\midrule 
\insta & 10080 & 300000\\
\zipf, \uniform & 10000 & 50000\\
\servicet & 4613 & 50000\\
\serviceo, \servicep & 4613 & 100000 \\
\bottomrule
\end{tabular}
\vspace{-1em}
\end{table}
We evaluate cube queries on our datasets with categorical dimension columns,
and partitioned them along four dimension columns
with parameters summarized in Table~\ref{tab:cube_setup} and
total space limit set to provide roughly consistent query error across the
datasets.
For each of these datasets we evaluate on a default query workload where each
dimension has an independent $p=.2$ probability of being included as a filter, 
and if selected the dimension value is chosen uniformly at random.

In Figures \ref{fig:cube_freq_global} and \ref{fig:cube_quant_global}
 we show the average relative error for frequency and quantile queries
over 10000 random cube queries drawn over the different dataset workloads.
We see that, on average, \storyboard outperforms alternative summarization
techniques that allocate equal space to each segment, as well as 
uniform sampling techniques that optimize sample size allocation.

\subsection{Varying Parameters}
\label{sec:eval_coopdesign}
Now we vary different system and summarization parameters to see their
impact on accuracy, confirming that \storyboard is able to provide
improved accuracy under a variety of conditions.

\subsubsection{System Design}
\label{sec:eval_sys_parameters}
\label{sec:cube_lesion}
The \storyboard system depends on a number of parameters that
go beyond a single summary.
In this section we will show how accuracy varies with the accumulator size $s_A$,
the size of cube aggregations, and the presence of each of the optimizations
storyboard uses for data cubes.

\minihead{Finite Accumulator}
\begin{figure}
    \centering
    \includegraphics[width=\columnwidth]{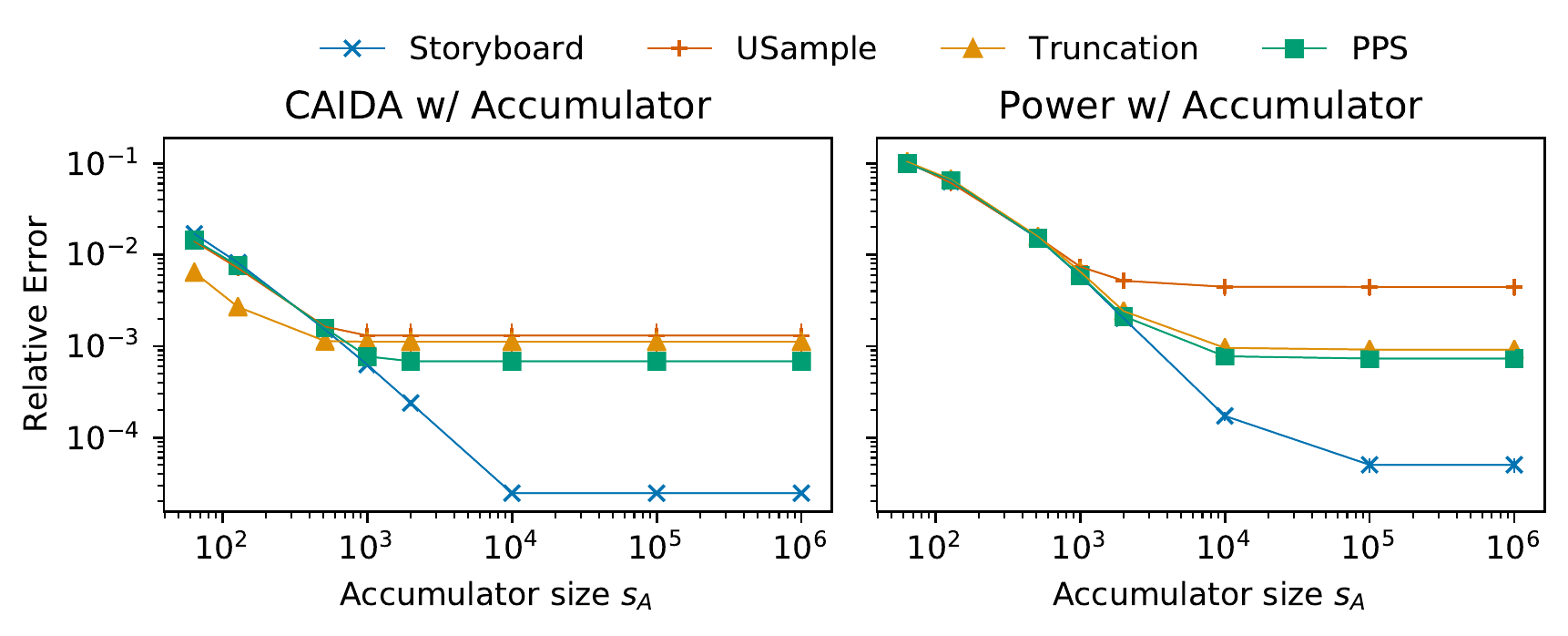}
    \caption{Query error as we vary the size of the accumulator $s_A$. For the large accumulators used in practice there is negligible additional error from the accumulator.}
    \label{fig:linear_acc}
\end{figure}
As described in Section~\ref{sec:sys-queries}, \storyboard accumulates precise
frequency and ranks estimates from the summaries for point queries $\hat{g}(x)$, and 
accumulates summaries into a large accumulator $A$ for quantile and heavy hitters
queries.
In our evaluations thus far we have measured the maximum point query error, which
bounds the maximum quantile or heavy hitter error as well (Section~\ref{sec:sys_error}).
The accumulator $A$ introduces an additional approximation error $\er^{\text{(A)}} = 1/s_A$
which is negligible as $s_A \rightarrow \infty$.

In Figure~\ref{fig:linear_acc} we illustrate how using accumulators of different sizes,
affects final query accuracy on the \power and \caida datasets.
For each accumulator size, we measure the error after accumulating 100 random
interval aggregations spanning $k=512$ segments.
For the accumulators here we use SpaceSaving \cite{metwally2005efficient}
for frequency queries
and a streaming implementation of \pps (\texttt{VarOpt} \cite{cohen2011structure}) for quantiles.
$\er^{\text{(A)}}$ quickly goes to $0$ as $s_A \rightarrow \infty$, so with
at least 10 megabytes of memory available for $s_A$ the additional error
is negligible.

\minihead{Cube Query Spans}
As seen already for interval queries, \storyboard improves the query error for
queries combining results from multiple segment summaries.
We can confirm this for cubes by looking at query error broken down
by the number of dimensions in each query filter condition.
Queries that filter on fewer columns will combine results from more segments.
In Figure~\ref{fig:cube_freq_query} we compare the error for queries that filter on
different numbers of dimensions on the \uniform and \zipf cube workloads.
\storyboard reduces the error for common queries that filter on zero or one dimension.
As a tradeoff \storyboard incurs higher error than other methods for rarer queries with
three or more filters.
For a workload where queries with 3 or more filters are much less common
than queries with 0 or 1 filters, this tradeoff is useful, and
is configurable based on the user specified workload.

\minihead{Cube Optimizer Lesion Study}
In Figure~\ref{fig:cube_freq_lesion} we show how the optimizations \storyboard (\sbopt)
uses for summarizing data cubes all play a role in providing high query accuracy
by removing individual optimizations on the \zipf dataset.
We experiment with removing the size optimizations (\sbopt (-Size)) and bias optimizations (\sbopt (-Bias)),
and try replacing \pps summaries with uniform random samples (\sbopt (-PPS)).
When, size optimization or bias optimization are removed, error increases, and
similarly error increases when \pps summaries are replaced with uniform random samples.

\minihead{Cube Workload Specification}
We also evaluate how \storyboard accuracy depends on precise workload specification
by constructing \storyboard instances configured for incorrectly specified workloads.
Rather than the true $p=0.2$ probability of including a dimension in the cube filter,
we try optimizing cubes for $p=0.05$ in Work1 and $p=0.50$ in Work2.
As included in Figure~\ref{fig:cube_freq_lesion},
in both cases error remains below existing cube construction methods.
Interestingly, accuracy improves for Work1, indicating our optimization is not tight.

\begin{figure}
    \centering
    \includegraphics[width=\columnwidth]{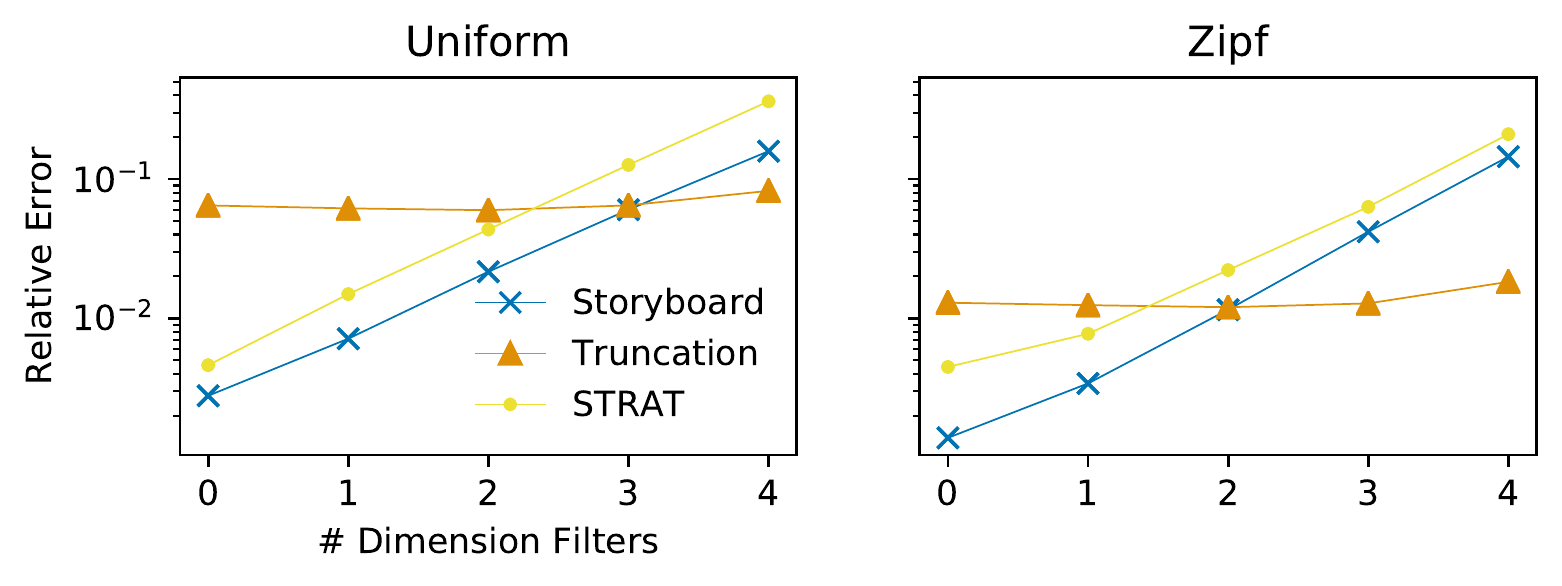}
    \caption{Query error broken down by number of dimension filters in a query. \storyboard achieves lower error on queries that have fewer filters and aggregate more segments.}\label{fig:cube_freq_query}
\end{figure}

\begin{figure}
    \centering
    \includegraphics[width=\columnwidth]{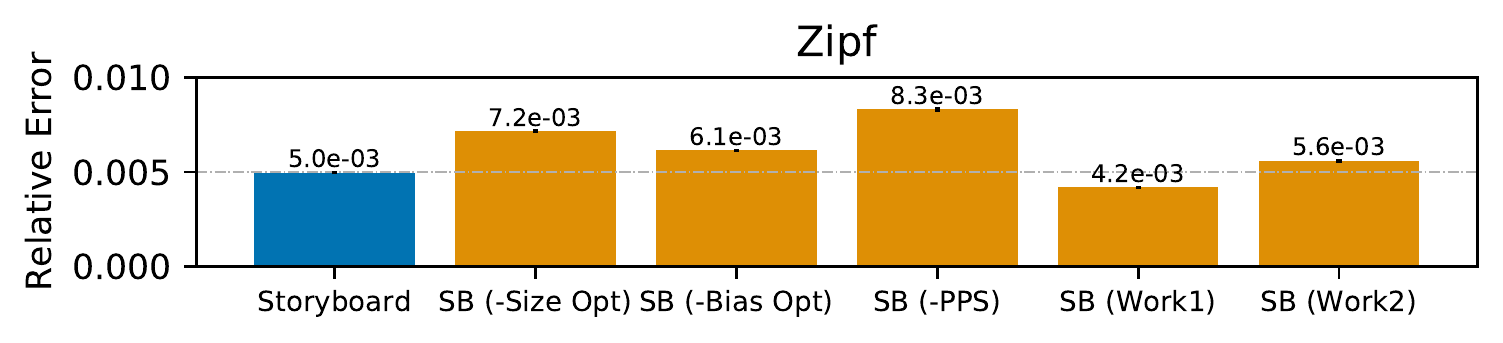}
\caption{Lesion Study on \zipf cube optimizations. Removing any component reduces accuracy, though
adjusting the workload parameter slightly improves accuracy.}
\label{fig:cube_freq_lesion}
\end{figure}

\minihead{Interval Length Specification}
For interval aggregations users specify a maximum expected interval length $k_T$.
In Figure~\ref{fig:linear_freq_lookback} we show the relative error for 20 random queries of length $k=64$ as we vary $k_T$.
All values $k_T\geq 64$ achieve good error and setting $k_T$ much larger does not negatively affect results.
In practice accuracy is also robust to different values of $k_T$ as long as it is
conservatively longer than the expected queries.
\begin{figure}[t]
    \centering
    \includegraphics[width=\columnwidth]{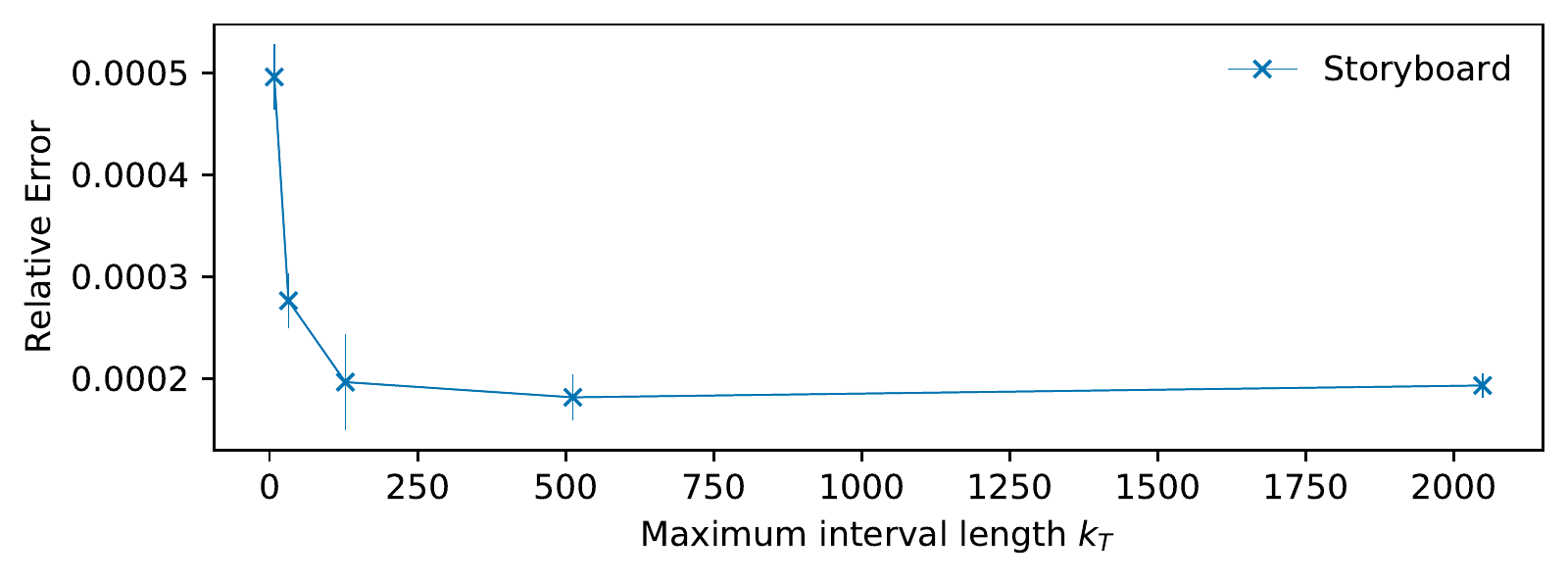}
    \caption{Error as we vary the maximum interval length parameter $k_T$. Overestimating $k_T$ does not significantly change the quality of results. \label{fig:linear_freq_lookback}}
\end{figure}

\subsubsection{Summary Design}
\label{sec:eval_summ_design}
Now we will examine how \storyboard's 
Cooperative and \pps summaries perform as individual segment summaries.
The experiments below are run on the \caida dataset for interval aggregations.

\minihead{Space Scaling}
In Figure~\ref{fig:linear_size} we vary the space available to summaries for different interval lengths,
confirming that like other state of the art summaries and sketches Cooperative and \pps summaries
provide local segment error that
scales inversely proportional to the space given, and maintain their accuracy under a wide range
of summary sizes.

\minihead{Hierarchical base $b$}
Although \dyadic summaries are parameterized by a base $b$, 
in Figure~\ref{fig:linear_freq_dyadic_base} we show that different values for $b$ do not noticeably improve performance.
Although there are improvements in optimizing $b$ when merging small numbers ($k<10$) of summaries,
the difference is less than $10\%$ for larger aggregations.

\begin{figure}[t]
    \centering
    \includegraphics[width=\columnwidth]{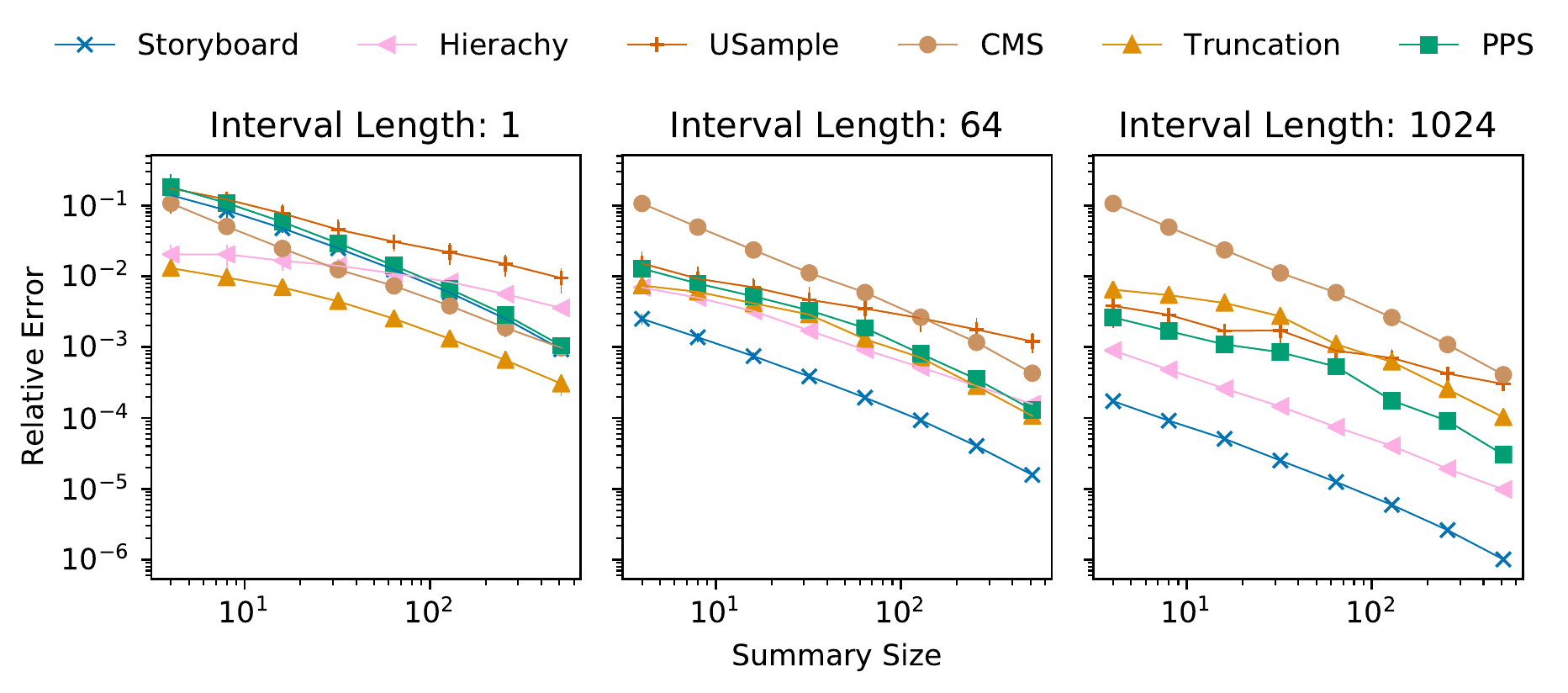}
    \caption{Query error as summary size changes. 
    Cooperative summaries, like state of the art, have error $\er = O(1/s)$}\label{fig:linear_size}
\end{figure}
\begin{figure}[t]
    \centering
    \includegraphics[width=\columnwidth]{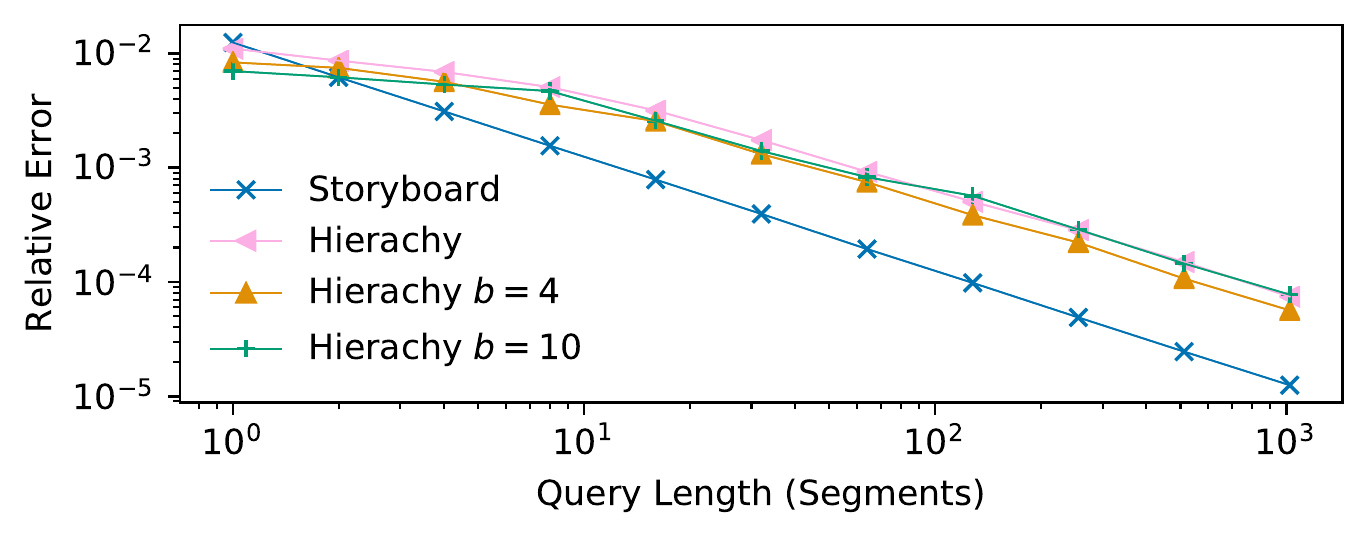}
    \caption{\dyadic summary accuracy for different bases $b$. $b$ does not have a large impact when accumulating across multiple summaries. \label{fig:linear_freq_dyadic_base}}.
    \vspace{-1em}
\end{figure}
\section{Related Work}
\label{sec:related}

\minihead{Precomputing Summaries}
A number of existing approximate query processing (AQP) systems make use of
precomputed approximate data summaries.
An overview of these ``offline'' AQP systems can be found in \cite{Li2018AQPSurvey},
and they are an instance of the \preagg systems described in \cite{peng2018aqp}.
Like data cube systems they materialize partial results \cite{harinarayan1996cubes}, but can support more complex query functions not captured by simple totals.
Another class of systems use ``online'' AQP \cite{Hellerstein1997OLA,Budiu2019hillview,rabkin2014agg} and
provide different latency and accuracy guarantees by computing
approximations at query-time.

We are particularly motivated by Druid \cite{yang2014druid,druidblog} 
and similar offline systems \cite{IM2018Pinot} which aggregate over 
query-specific summaries for disjoint segments of data.
However, these systems use mergeable summaries as-is, and do not optimize
for improving accuracy under aggregation or take advantage of additional
memory at query time to accumulate results more precisely.
The authors in \cite{Yi2014Indexing} apply hierarchical strategies
to maintain summary collections for interval 
queries but like mergeable summaries maintain do not reduce error when
combining summaries.
Systems like BlinkDB \cite{agarwal2013blinkdb}, STRAT \cite{chaudhuri2007strat}, and AQUA \cite{acharya1999aqua} maintain random stratified samples to support general-purpose queries.
Our choice of minimizing mean squared error over a workload
follows the setup in STRAT \cite{chaudhuri2007strat}.
However, individual simple random samples are not as accurate as specialized
frequency or quantile summaries \cite{mozafari2015handbook}.

Techniques for summarizing hierarchical intervals \cite{basat2018interval} are complementary,
but incur additional storage overhead making them less accurate than Cooperative summaries
and scale poorly to cubes with multiple dimensions \cite{qadarji2013hierarchical}.

\minihead{Streaming and Mergeable Summaries}
Many compact data summaries are developed in the streaming literature \cite{greenwald2001space,misra1982gries,cormode2005countmin,karnin2016optimalquant}, including summaries
for sliding windows \cite{arasu2004window}.
However, they assume a different system model than \storyboard provides.
The standard streaming model generally considers queries with working memory limited during summary construction \cite{muthukrishnan2005data}.
Mergeable summaries \cite{agarwal2012mergeable} allow combining
multiple summaries but require that intermediate results take up no more space
than the inputs, and thus merely maintain relative error under merging.
Other work targeting \preagg systems has focused on 
improving summary update and merge runtime performance \cite{gan2018moments,masson2019ddsketch} 
rather than improving the accuracy of merged summary results.

\minihead{Other Summarization Models}
The \storyboard model, where more memory is available for construction and aggregation
than for storage, is closer to the model used in non-streaming settings including
discrepancy theory and communication theory.

Coresets and $\epsilon$-approximations are data structures for approximate queries
that allow more resource-intensive precomputation and aggregation \cite{philips2017coreset}. 
$\epsilon$-approximations are part of discrepancy theory 
which attempts to approximate an underlying distribution with proxy samples \cite{Chazelle2000discrepancy}.
We draw inspiration from discrepancy theory to manage error accumulation
in our cooperative summaries, especially the results in \cite{spencer1977balancing} 
which pioneered the use of the $\cosh$ cost function.
Other work in this area minimize error accumulation along multiple dimensions \cite{philips2008algsterrain}.
However, we are not away of coreset or $\epsilon$-approximations that allow for
complex queries \storyboard supports: quantiles and item frequencies over multiple data segments. 
and cube aggregations.
In particular, existing work supporting range queries \cite{philips2008algsterrain} do not provide
per-segment local guarantees.

Work in communication theory and distributed streaming assume a setting where sending
summaries over the network is a bottleneck equivalent to storage costs limits in \storyboard.
There is existing work analyzing how multiple random samples can be combined
to reduce aggregate error in this setting \cite{huang2011optimalsamp,zhao2006distributedice}.
However, in communication theory the samples are constructed per-query, while
\storyboard precomputes summaries that can be used for arbitrary future queries.
Furthermore the random samples are not as space efficient as cooperative summaries.

Related techniques in differential privacy \cite{qadarji2013hierarchical,cormode2019rangeprivacy} and
matrix rounding \cite{doerr2006matrix} consider approximate representations of data segments
for the purposes of privacy, but do not explicitly optimize for space or support heavy hitters
and quantile queries.
\section{Conclusion}
\label{sec:conclusion}
When aggregating multiple precomputed summaries, \storyboard
optimizes and accumulates summaries for reduced query error.
It does so by taking advantage of additional memory resources
at summary construction and aggregation while targeting a common
class of structured frequency and quantile queries.
This system can thus efficiently serve a range of monitoring
and data exploration workloads.
Extensions to other query and aggregation types are a rich area for future work.
% For future work, we hope to extend
% cooperative summaries to work with additional aggregation types, 
% and integrating elements from hierarchical summarization techniques \cite{basat2018interval,harinarayan1996cubes,qadarji2013hierarchical} into the system.
\iftoggle{arxiv}{}{
	\begin{appendix}
	\section{Cooperative Summary Proofs}
\label{sec:freq_pot_proof}
\minihead{Lemma~\ref{lem:coopfreqpot}}
\begin{proof}
Recall that we have a segment 
$$\DD_t = \{x_1 \mapsto \delta_1, \dots, x_r \mapsto \delta_r\}.$$
Let $H$ be the set of local heavy hitters $H = \{x_i : \delta_i \geq h\}$
and let $U' = U \setminus H$ be the remaining items.
We can decompose our summary as $S_t = S_H \cup S_V$
where $V = S_t \setminus H$.
\begin{align}
S_H &= \{x_i \mapsto \delta_i : x_i \in H\} \\ 
S_V &= \{x_i \mapsto \min(\ea_{t-1}(x_i)+\delta_i, rh) : x_i \in V\}.
\end{align}
This keeps $\ea_t(x) \geq 0$ across segments, i.e. our estimates are always underestimates.

Let $G = L_t - L_{t-1}= \sum_{x_i \in U} \left[\pot(\ea_{t}(x_i)) - \pot(\ea_{t-1}(x_i))\right]$
where $\pot(z) = \exp(\alpha z)$.
For heavy hitters $\ea_{t}(x_i) = \ea_{t-1}(x_i)$ so they do not change the cumulative cost $L_t$.
\begin{align*}
G &= \sum_{x_i \in V} \left[\pot(\max(\ea_{t-1}(x_i) + \delta_i - rh,0)) - \pot(\ea_{t-1}(x_i))\right] \\
 & \quad + \sum_{x_i \in U' \setminus V} \left[\pot(\ea_{t-1}(x_i) + \delta_i) - \pot(\ea_{t-1}(x_i))\right] 
\end{align*}
Simplifying using $\max(0,y) = y + (0-y)1_{y\leq 0}$ and $\pot(x+y)=\pot(x)\pot(y)$:
\begin{align*}
G &= \sum_{x_i \in U' \setminus V} \pot(\ea_{t-1}(x_i) + \delta_i)\left[1 - \pot(-\delta_i)\right] \\
 & \quad + \sum_{x_i \in V} \pot(\ea_{t-1}(x_i) + \delta_i)\left[\pot(-rh) - \pot(-\delta_i)\right] \\
 & \quad + \sum_{x_i \in V} \left[\pot(0) - \pot(\ea_{t-1}(x_i) + \delta_i - rh)\right]\cdot 1_{\ea_{t-1}+\delta_i \leq rh}
\end{align*}

For non-heavy hitters, Algorithm~\ref{alg:greedy_freq} selects items in $V$ with the highest $\ea_{t-1}(x_i) + \delta_i$.
If we let $\ell = \arg\min_{x_i \in V} \ea_{t-1}(x_i) + \delta_i$ then
\begin{align*}
\forall x_i \in V&\quad \ea_{t-1}(x_{\ell}) + \delta_{\ell} \leq \ea_{t-1}(x_i) + \delta_{i} \\
\forall x_i \in U' \setminus V&\quad \ea_{t-1}(x_{\ell}) + \delta_{\ell} \geq \ea_{t-1}(x_i) + \delta_{i}.
\end{align*}
Technical but standard applications of the inequalities
$\pot(x) \geq 1+ax$, and $\pot(x) \leq 1 + \alpha x + \alpha^2 x^2/2$ for 
$x \leq 0$ yields:
\begin{align*}
G \leq& \pot(\ea_{t-1}(x_\ell) + \delta_\ell) |V| \left[
\alpha h - \alpha h r + \alpha^2 h^2 r^2/2 
\right] + \alpha r h |V|
\end{align*}
$|V| \leq s$ and $h \leq |\DD_t|/s$ so when $\alpha \leq \frac{2}{h}\frac{r-1}{r^2}$, $G \leq \alpha r |\DD_t|$
\end{proof}

\label{sec:quant_cosh_proof}
\minihead{Lemma \ref{lem:quant_cosh}}
\begin{proof}
First note that the choice of which element $z_j$ is chosen from each chunk
for inclusion in the summary sample $S_t$ 
does not affect $\ea_t(x)$ for $x$ outside the
chunk $\DD_{t,j}$ so we can consider the choices independently.
This is because the selected element is assigned a proxy count equal to the
population of the whole chunk $h=|\DD_{t,j}|=|\DD_t|/s$.

Let $L_{t,j} \coloneqq \sum_{x_i \in \DD_{t,j}} \pot\left(\ea_t(x_i)\right)$ be total cost for chunk $j$.
Since Algorithm~\ref{alg:greedy_quant} selects a value $z$ that minimizes $L_t$, 
the final value for $L_{t,j}$ must be lower than any weighted average of the possible $L_{t,j}$ for
different choices of $x$.
\begin{align*}
L_{t,j} &\leq \sum_{z \in \DD_{t,j}} \frac{f_{\DD_t}(z)}{h} \left[
\sum_{x\in \DD_{t,j}} \pot\left(\ea_{t-1}(x) + r_{\DD_{t,j}}(x) - 1_{x \geq z}h\right)
\right]
\end{align*}
Abbreviate $p_x \coloneqq \frac{1}{h}r_{\DD_{t,j}}(x) = \frac{1}{h}\sum_{x_i \in \DD_{t,j}} \delta_i \cdot 1_{x_i \leq x}$.
Switching the order of summation gives:
\begin{align*}
L_{t,j} &\leq \sum_{x\in \DD_{t,j}}[
p_x \pot(\ea_{t-1}(x) + hp_x - h) 
\\ & \qquad 
+ (1-p_x)\pot(\ea_{t-1}(x) + hp_x) ]
\end{align*}
Now we can make use of Lemma~\ref{lem:coshineq} below to simplify
\begin{align*}
L_{t,j} 
% &\leq \sum_{x\in U_j}\exp(\alpha^2 h^2/2)\pot(\ea_{t-1}(x)) \\
	&\leq \exp(\alpha^2 h^2/2) L_{t-1,j}
\end{align*}

Finally, since $L_t = \sum_{j=1}^{s} L_{t,j}$ we have the lemma.
\end{proof}

Lemma~\ref{lem:coshineq} can be proven using the cosh angle addition 
formula and Taylor expansions.
\begin{lemma}
For $0 \leq p \leq 1$ and $t \geq 0$
\begin{equation}
p \cosh(x + t(p-1)) + (1-p)\cosh(x+tp) \leq \exp(t^2/2)\cosh(x)
\label{eqn:coshineqq}
\end{equation}
\label{lem:coshineq}
\end{lemma}
\iftoggle{arxiv}{
\begin{proof}
We abbreciate $\cosh,\sinh$ as $c,s$ and the left hand side of 
Equation~\ref{eqn:coshineqq} as $LHS$.
Using the angle addition formula:
\begin{align*}
LHS &= p\left[c(x)c(t(p-1)) + s(x)s(t(p-1))\right] \\
	&\quad + (1-p)\left[c(x)c(tp)+s(x)s(tp)\right]
\end{align*}
Then since $s(x)\leq c(x)$:
\begin{align*}
LHS &\leq pc(x)\left[c(t(p-1)) + s(t(p-1))\right] \\
	&\quad + (1-p)c(x)\left[c(tp)+s(tp)\right] \\
	&= c(s)\left[p\exp(t(p-1)) + (1-p)\exp(tp)\right]
\end{align*}

We now consider two cases: $t<2$ and $t\geq 2$.

If $t<2$, we expand out taylor series to get that:
\begin{align*}
p\exp(t(p-1)) + (1-p)\exp(tp) &\leq 1 + \frac{t^2}{2}(p(1-p))\cdot 3 \\
&\leq 1+t^2/2 \leq \exp(t^2/2)
\end{align*}
If $t\geq 2$ then 
\begin{align*}
p\exp(t(p-1)) + (1-p)\exp(tp) &\leq \exp(tp) \\
&\leq \exp(t) \leq \exp(t^2/2)
\end{align*}

In either case we can conclude that:

\begin{align*}
LHS &\leq \cosh(x)\exp(t^2/2)
\end{align*}

\end{proof}
}{}

\iftoggle{arxiv}{
\subsection{Error Lower Bounds}
\label{sec:error_lowerbounds}
In this section we will provide details on an adversarial dataset
for which no online selection of items for a counter-based
summary can achieve better than absolute $\ea = \Omega(\log{k})$ error
for item frequency queries.
\begin{theorem}
There exists a sequence of $k=2^{h+1}$ data segments $\DD_i$ consisting
of $|\DD_i| = 2s$ item values each such that for all possible selections of $s$ items for counter-based summaries $S_i$,
$\exists x. |f_{\DD_i,\dots,\DD_k}(x) - \hat{f}_{S_i,\dots,S_k}(x)| \geq h$.
\end{theorem}
\begin{proof}
Consider a universe of item values $U=1,\dots,2s2^{h}$.
For $i=1,\dots,2^h$ let $\DD_i = \{2s(i-1)+1,\dots,2si\}$ where
each item occurs at most once.
Since each summary $S_i$ can only store $s$ item values, 
there must be a set of $s2^h$ items ($U_{1}$) that are not stored in any summary, but that have appeared at least once in the data.
Now let the next $2^{h-1}$ data segments $\DD_i$ for $i\geq 2^h+1$
constain $2s$ distinct item values each from $U_{1}$.
Again, since each summary can only store $s$ item values
now there must be a set of $s2^{h-1}$ items ($U_{2}$) that are not
stored in any summary, but that have appeared twice in the data.
This repeats for for increasing $U_{i}$: at each stage $U_i$ we have
data segments come in that contain only items the summaries have not
been able to store, until we have at least one item not stored in any summary but that has appeared $h+1$ times in the data.
\end{proof}
}{}
	\end{appendix}
}

\iftoggle{arxiv}{
\section*{Acknowledgments}
{
\small
This research was made possible with feedback and assistance from our collaborators at Microsoft including Atul Shenoy, Asvin Ananthanarayan, and John Sheu, as well as collaborators at Imply including Gian Merino.
We also thank members of the Stanford Infolab
for their feedback on early drafts of this paper.
This research was supported in part by affiliate members and other supporters of the Stanford DAWN project---Ant Financial, Facebook, Google, Infosys, NEC, and VMware---as well as Toyota Research Institute, Northrop Grumman, Amazon Web Services, Cisco, and the NSF under CAREER grant CNS-1651570.
}
\bibliographystyle{ACM-Reference-Format}
}{
\bibliographystyle{abbrv}
}
\Urlmuskip=0mu plus 1mu
\bibliography{paper} 

\iftoggle{arxiv}{
	\begin{appendix}
	
	\end{appendix}
}{}

\end{document}